\documentclass[twocolumn,superscriptaddress,showpacs,amsmath,amstex,amssymb,citeautoscript,floatfix]{revtex4-2}
\pdfoutput=1
\usepackage[english]{babel}
\usepackage{letltxmacro}
\usepackage{latexsym}
\LetLtxMacro{\ORIGselectlanguage}{\selectlanguage}
\makeatletter
\DeclareRobustCommand{\selectlanguage}[1]{%
  \@ifundefined{alias@\string#1}
    {\ORIGselectlanguage{#1}}
    {\begingroup\edef\x{\endgroup
       \noexpand\ORIGselectlanguage{\@nameuse{alias@#1}}}\x}%
}
\newcommand{\definelanguagealias}[2]{%
  \@namedef{alias@#1}{#2}%
}
\makeatother
\definelanguagealias{en}{english}
\definelanguagealias{English}{english}
\usepackage{graphicx}
\usepackage{amsmath}
\usepackage{amsfonts}
\usepackage{amssymb}
\usepackage{amsthm}
\usepackage{cancel}
\usepackage{mathtools}
\usepackage{bm}
\usepackage{color}
\usepackage[percent]{overpic}
\usepackage{soul} 
\usepackage{amssymb}
\usepackage{wasysym}
\usepackage{dsfont}
\usepackage{float}
\usepackage{physics}
\usepackage{hyperref}
\hypersetup{
  colorlinks   = true, %Colours links instead of ugly boxes
  urlcolor     = blue, %Colour for external hyperlinks
  linkcolor    = blue, %Colour of internal links
  citecolor   = red %Colour of citations
}
\usepackage[mathscr]{euscript}

\newtheorem{lemma}{Lemma}
\newtheorem{theorem}{Theorem}
\newtheorem*{theorem*}{Theorem}
\newtheorem*{lemma*}{Lemma}

\def\L{\left} 
\def\R{\right}
\newcommand{\id}{\mathbb{I}}
\newcommand{\swap}{\mathcal{S}}
\newcommand{\Ebb}{\mathbb{E}}
\newcommand{\pdiagz}{p_\text{avg}(z)}
\newcommand{\pdiag}[1]{p_\text{avg}({#1})}
\newcommand{\FXEBd}{F_\text{XEB,d}}
\newcommand{\Fidd}{F_\text{id,d}}
\newcommand{\Freld}{F_e}

\newcommand{\be}{\begin{equation}}
\newcommand{\ee}{\end{equation}}
\newcommand{\bea}{\begin{eqnarray}}
\newcommand{\eea}{\end{eqnarray}}

\usepackage{graphicx}
\usepackage{verbatim}
\usepackage[normalem]{ulem}

\newcommand{\Caltech}{California Institute of Technology, Pasadena, CA 91125, USA}
\newcommand{\MIT}{Center for Theoretical Physics, Massachusetts Institute of Technology, Cambridge, MA 02139, USA}
\begin{document}
\title{
Benchmarking quantum simulators using ergodic quantum dynamics
}
\author{Daniel K.~Mark}
\affiliation{\MIT}
\author{Joonhee Choi}
\affiliation{\Caltech}
\author{Adam L.~Shaw}
\affiliation{\Caltech}
\author{Manuel Endres}
\affiliation{\Caltech}
\author{Soonwon Choi}\email{soonwon@mit.edu}
\affiliation{\MIT}

\begin{abstract}
We propose and analyze a sample-efficient protocol to estimate the fidelity between an experimentally prepared state and an ideal target state, applicable to a wide class of analog quantum simulators without advanced spatiotemporal control.
Our protocol relies on universal fluctuations emerging from generic Hamiltonian dynamics, that we discover in the present work. It does not require fine-tuned control over state preparation, quantum evolution, or readout capability, while achieving near optimal sample complexity: a percent-level precision is obtained with $\sim 10^3$ measurements, independent of system size. Furthermore, the accuracy of our fidelity estimation improves exponentially with increasing system size.
We numerically demonstrate our protocol in a variety of quantum simulator platforms including quantum gas microscopes, trapped ions, and Rydberg atom arrays.
We discuss applications of our method for tasks such as multi-parameter estimation of quantum states and processes.
\end{abstract}

\maketitle

\textit{Introduction.---}Recent advances in quantum technology have opened new ways to probe quantum many-body physics, leading to the first observations of novel phases of  matter~\cite{greiner2002quantum,aidelsburger2013realization,jotzu2014experimental,gross2017quantum,ebadi2021quantum}, quantum thermalization~\cite{kaufman2016quantum,tang2018thermalization,chen2021observation}, and non-equilibrium phenomena~\cite{neyenhuis2017observation,zhang2017observation,choi2019probing,rubio2020floquet,peng2021floquet,mi2021time}.
However, in order to advance to the stage where quantum devices produce highly accurate data, it is important to quantify the performance of said devices. One method to do so is quantum device benchmarking~\cite{eisert2020quantum}---verifying that a device accurately produces a state $\rho$ close to the desired state $\ket{\Psi}$, in the presence of imperfections and noise, measured by the fidelity $F = \bra{\Psi}\rho\ket{\Psi}$. A high fidelity certifies that any property of the prepared state is close to that of the target state~\cite{nielsen_chuang_2010}, hence is widely used in theory to quantify the goodness of state preparation. Experimentally measuring the fidelity is important for building, characterizing, and improving increasingly complex and precise systems.

Several methods to benchmark quantum devices have been proposed. A na\"ive approach is to perform quantum state tomography~\cite{nielsen_chuang_2010,cramer2010efficient,gross2010quantum,christandl2012reliable}, in which an experimental state is fully characterized by measurements in many different bases.
This approach, however, is impractical even for relatively small systems as it requires prohibitively many measurements. 
Alternatively, recent proposals pointed out that one can directly estimate the fidelity with a small number of measurements in randomly chosen bases~\cite{flammia2011direct,daSilva2011practical,aaronson2019shadow,huang2020predicting,elben2019statistical,elben2020cross,liu2021benchmarking,elben2022randomized}.
These methods rely on implementing highly engineered quantum gates that satisfy certain statistical properties and are not readily applicable to quantum devices with limited controllability. Other existing benchmarking protocols require sophisticated controls and are challenging to implement~\cite{ohliger2013efficient,brydges2019probing,boixo2018characterizing,arute2019quantum,li2020hamiltonian,kokail2021entanglement,gluza2021recovering,hu2021classical,choi2021emergent}. In particular, we emphasize that analog quantum simulators are typically designed to realize specific forms of many-body Hamiltonians and lack the ability to implement arbitrary unitary operations. Thus, it remains an outstanding challenge to develop a general benchmarking method with minimal requirements on hardware capability.

\begin{figure}
    \centering
    \includegraphics[width=\columnwidth]{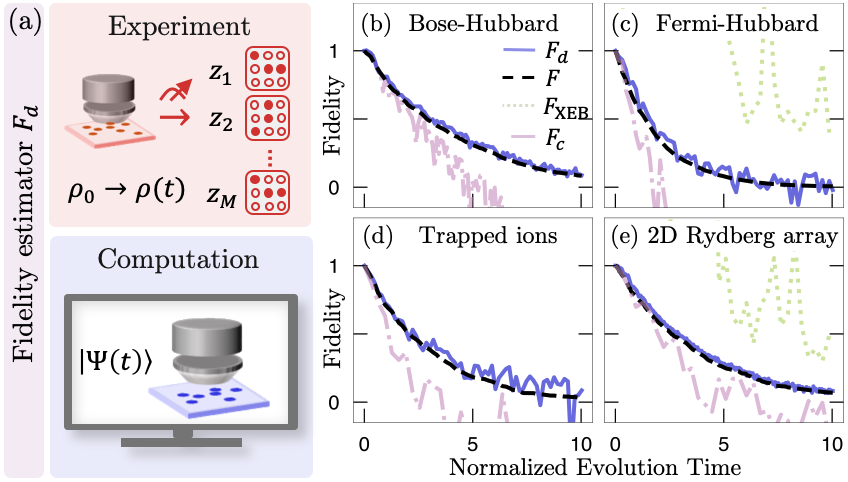}
    \vspace{-15pt}
    \caption{
    (a) Schematic of our benchmarking protocol. The fidelity estimator $F_d$ is evaluated from experimental snapshots $\{z_1,\dots, z_M\}$ of a state $\rho(t)$ obtained after quench dynamics, compared against classical computation of the ideal state $\ket{\Psi(t)}$ in the absence of error (see Table~\ref{tab:protocol}). 
    (b-e) Numerical demonstrations. $F_d$ closely tracks the fidelity decay over evolution time normalized in units of Rabi frequency or tunnelling strength (black dashed) between noisy and ideal quench dynamics in a wide class of analog simulators, including 1D Bose-Hubbard, integrable 1D Fermi-Hubbard, 1D trapped-ion, and 2D Rydberg array models at finite effective temperature, see SM~\cite{SM} for details. Previously proposed benchmarks $F_\textrm{XEB}$~\cite{boixo2018characterizing} [green dotted, out of scale in (b,d)] or $F_c$~\cite{choi2021emergent} (purple dot-dashed) fail to estimate $F$ for these systems. 
}
\label{fig:Fig1}
\end{figure}

In this work, we propose and analyze a benchmarking protocol that requires minimal experimental control: one prepares an initial state, time-evolves it under a natural Hamiltonian of the system, and performs measurements in a fixed basis (Fig.~\ref{fig:Fig1}).
We show that, with appropriate data-processing (enabled by classical computation), this simple experiment gives an estimate for the fidelity $F$---encapsulating the combined effects of errors in state preparation, quench evolution, and readout~\footnote{More accurately, the fidelity is a quantity defined between two quantum states and does not include readout errors. Our benchmark estimates the state fidelity when readout errors are negligible.}---with a small number of measurements. Most importantly, our method works for generic quench dynamics far from fine-tuned cases, including: at finite effective temperatures, in the presence of symmetries, and in non-qubit based systems such as itinerant particles on optical lattices, making it suitable for a wide class of existing platforms.

The key behind our approach is our discovery of universal statistical fluctuations in the measurement outcome distributions $p(z)$ that arise from generic quantum dynamics (Fig.~\ref{fig:FigPT}).
Previously, such universal fluctuations in $p(z)$ were only known to occur in ideal, controlled dynamics such as random unitary circuits (RUCs), where $\{p(z)\}$ approximately follows the Porter-Thomas distribution~\cite{porter1956fluctuations,arute2019quantum}. Leveraging our discovery and classical computation, we design a novel statistic: a real number $f(z)$ associated to every measurement outcome $z$ such that its average over experimentally obtained samples, $\hat{F}_d \equiv \langle f(z) \rangle_\textrm{exp}$, converges quickly to the many-body fidelity $F$. In other words, $\hat{F}_d$ is a computationally-assisted, efficiently measurable observable~\cite{garratt2022measurements,lee2022decoding} that estimates the fidelity.

The ability to estimate fidelity serves as a foundation for two tasks: (i) \emph{target state benchmarking}, where the overlap between an experimentally prepared state and a pure target state is measured via a high-fidelity quench time evolution, and (ii) \emph{quantum process benchmarking}, in which the fidelity decay of quench dynamics is monitored over the course of evolution.

\textit{Protocol.---}We focus on describing and numerically demonstrating our protocol, before returning to why it works. Our benchmarking method consists of three steps: experiment, computation, and data processing [Table~\ref{tab:protocol} and Fig.~\ref{fig:Fig1}(a)]. The initial state of our protocol can either be an easy-to-prepare state or a more complex state that one wishes to benchmark. After quench evolution for a fixed time $t$, the experimental state $\rho(t)$ is measured in any fixed basis $\{\ket{z}\}$. Convenient choices of $\{\ket{z}\}$ include the set of bitstrings in two-level (qubit) systems or real-space particle number configurations in quantum gas microscopes. Repeating the state preparation and measurement $M$ times, one obtains measured configurations $\{z_1,$\,$\dots$\,$, z_M\}$, with each $z_i$ sampled from the distribution $q(z,t)$\,$\equiv$\,$\bra{z}\rho(t)\ket{z}$.
Our protocol estimates the fidelity by using a small number of samples to compare the empirical distribution $q(z,t)$ against a theoretical, target distribution $p(z,t)$. By classical computation, we obtain $p(z,t)$ and its infinite-time average $\pdiagz \equiv \lim_{T \rightarrow \infty} \frac{1}{T}\int_0^T p(z,t) dt$. In practice, one may average over a finite duration $T$ as an approximation, at the expense of slightly larger statistical errors. 
Then, we evaluate the rescaled outcome probabilities $\tilde{p}(z,t) \equiv p(z,t)/\pdiagz$ and the normalization factor $\mathcal{Z}(t) \equiv \sum_z \pdiagz \tilde{p}(z,t)^2$.

The classical computation determines our statistic $f(z)$, while the experimental samples determine which outcomes $z$'s to use when evaluating the statistic. More specifically, we estimate the fidelity with the empirical average 
\begin{equation}
\hat{F}_d(t) = \langle f(z) \rangle_\textrm{exp}= \frac{1}{M} \sum_{i=1}^{M} 2\tilde{p}(z_i,t)/\mathcal{Z}(t) - 1~.
\end{equation}
This explicitly defines the statistic $f(z)$, that also depends on $\ket{\Psi_0}$, $H$, and $t$. In the limit $M\!\rightarrow\!\infty$, this converges to our benchmark $F_d(t) = 2\Big[\sum_z q(z,t) \tilde{p}(z,t)\Big]/\mathcal{Z}(t) - 1$~.
This benchmark can be understood as a weighted covariance between the empirical and ideal distributions. We show that $F_d(t)$ approximates the fidelity $F$ for a wide class of quantum systems, both for uncorrelated, infrequent incoherent errors, or local and weak global coherent errors~\cite{gao2021limitations,dalzell2021random,Noh2020efficientclassical,SM}, and rigorously prove our statement for isolated single errors and long evolution times.

\renewcommand{\arraystretch}{1.3}
\begin{table}[t!]
    \centering
    \vspace{-0.1in}
    \caption{Proposed benchmarking protocol.}
    \vspace{0.2in}
        \label{tab:protocol}
    \begin{tabular}{p{0.175in} p{3in}}
\hline
\multicolumn{2}{l}{\textbf{Experiment:}}\\
&1. Prepare an initial state $\rho_0$, which approximates a pure state $\ketbra{\Psi_0}{\Psi_0}$.\\
&2. Evolve the system under its natural Hamiltonian $H$ for a time $t$.\\
&3. Measure the evolved state $\rho(t)$ in a natural basis, obtaining configurations $\{z_j\}_{j=1}^M$.\\
\multicolumn{2}{l}{\textbf{Computation:} Classically compute}\\
&1. $p(z,t) \equiv \abs{\bra{z}\ket{\Psi(t)}}^2 = \abs{\bra{z}\exp(-iHt)\ket{\Psi_0}}^2$,\\
&2. $\pdiagz \equiv \lim_{T \rightarrow \infty} \frac{1}{T}\int_0^T p(z,t)~dt$,\\
&~~~~$~\tilde{p}(z,t) \equiv p(z,t)/\pdiagz$.\\
&3. $\mathcal{Z}(t) \equiv  \sum_z \pdiagz \tilde{p}(z,t)^2$.\\
\multicolumn{2}{l}{\textbf{Data processing:} Evaluate}\\
\multicolumn{2}{c}{$\hat{F}_d(t) \equiv  \frac{2}{M} \Big[\sum_{i=1}^M \tilde{p}(z_i,t) \Big]/\mathcal{Z}(t) - 1 \approx F_d(t)$~,}\\
\multicolumn{2}{p{3.175in}}{which approximates the fidelity $F=\bra{\Psi(t)}\rho(t)\ket{\Psi(t)}$.}\\
\hline
    \end{tabular}
\vspace{-0.2in}
\end{table}

Figure~\ref{fig:Fig1}(b-e) numerically demonstrates the use of our estimator for process benchmarking: tracking the decay of fidelity over time in four different quantum simulation platforms.
For each platform, we simulate an initial product state undergoing natural Hamiltonian dynamics in the presence of experimentally relevant errors~\cite{SM}.
We confirm that $F_d$ successfully traces the fidelity decay in regimes where previously proposed fidelity estimators $F_\textrm{XEB}$~\cite{boixo2018characterizing} and $F_c$~\cite{choi2021emergent} do not.
This is because $F_c$ and $F_\textrm{XEB}$ (reviewed in the SM) assume that $p(z,t)$ satisfy statistical properties (discussed below) that are in general not satisfied by natural Hamiltonian dynamics, e.g. $F_c$ requires the system to evolve at infinite effective temperature. We now turn to the underlying principles of our protocol: \textit{emergent universal statistics}, \textit{speckle based benchmarking}, and \textit{measurement-basis independence}.

\begin{figure}[h]
    \centering
    \includegraphics[width=\columnwidth]{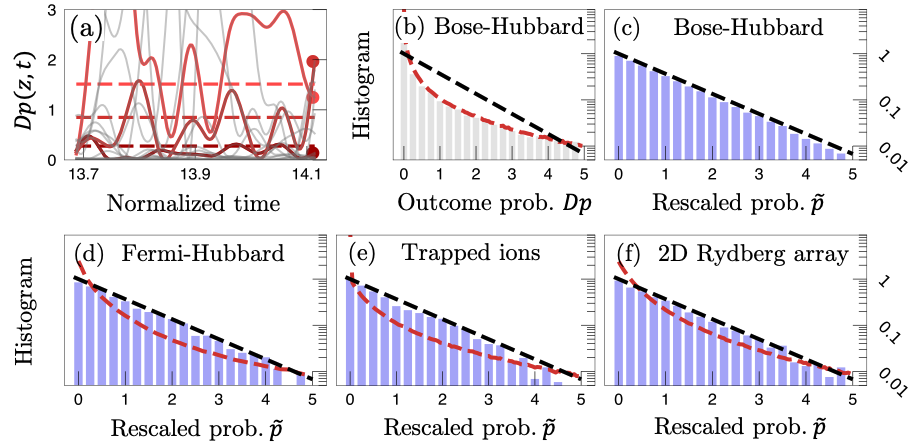}
    \vspace{-15pt}
    \caption{Emergent universal statistics. (a) During many-body Hamiltonian evolution of a pure state, the probability $p(z,t)$ of measuring an outcome $z$ fluctuates around its average value $p_\text{avg}(z)$ (dashed, three distinct $z$'s highlighted).
    Here, we consider a Bose-Hubbard model and present $p(z,t)$ rescaled by the Hilbert space dimension $D$.
    (b) The histogram of $p(z,t)$ over all $z$ at a fixed $t$ [red dots in (a)], is non-universal (red dashed), far from the Porter-Thomas (PT) distribution (black dashed). (c) Rescaling each $p(z,t)$ by $p_\text{avg}(z)$ yields $\tilde{p}(z,t)$, whose histogram follows the PT distribution.
    (d,e,f) The histograms of $\tilde{p}(z,t)$ (blue bars) follow the universal PT distribution in all models considered in this work, whereas those of the bare $p(z,t)$ rescaled by $D$ are non-universal (red dashed).}
\label{fig:FigPT}
\end{figure}

\textit{Emergent universal statistics.---}The statistical properties of $p(z)$ have been extensively studied in deep random unitary circuits (RUCs). When the output state of a typical deep RUC is measured, $p(z)$ is not perfectly uniform, but exhibits a \textit{speckle pattern}: over different $z$'s, $p(z)$ fluctuates about $1/D$ due to random interference in coherent quantum dynamics, with $D$ the Hilbert space dimension. While the details of the fluctuations---which $p(z)$'s are larger---sensitively depends on the particular choice of RUC, the statistical properties of $p(z)$ are universal. Specifically, the fraction of $p(z)$'s in a given interval $p(z)\in [x,x+dx]$ is given by the Porter-Thomas (PT) distribution: $P[p(z) = x]dx = \mu^{-1} \exp(-x/\mu) dx$ with mean $\mu = 1/D$. This enables the existing benchmarks $F_\textrm{XEB}$ and $F_c$. Specifically, they utilize the fact that the PT distribution has a second moment equal to two~\cite{boixo2018characterizing,arute2019quantum,choi2021emergent}. Previously, it was unclear under what conditions the PT distribution can arise, other than from RUCs and fine-tuned Hamiltonian dynamics. 

In fact, for generic time-independent Hamiltonian dynamics, the raw distribution $p(z)$ does not follow the PT distribution [Fig.~\ref{fig:FigPT}(a,b)]. This is due to the presence of energy conservation or symmetries, which causes systematic trends in $p(z)$ and distorts its distribution away from PT. For example, in any state with positive effective temperature, low-energy configurations are measured more frequently than high-energy ones~\cite{SM}. While previous work discovered PT distributions in certain \textit{local observables}~\cite{choi2021emergent}, this work concerns global observables under general conditions such as finite effective temperature.

Our key insight is that the systematic trends in $p(z)$ can be removed simply by rescaling $p(z)$ with its time-averaged value, leaving only random relative fluctuations $\tilde{p}(z) \equiv p(z)/\pdiagz$ that follow the PT distribution with mean $\mu = 1$ [Fig.~\ref{fig:FigPT}(c-f)].

\begin{theorem*} (informal)
Consider an initial state $\ket{\Psi_0}$, evolved for time $t$ under a Hamiltonian $H$ satisfying the \textup{$k$-th no-resonance condition} for a large integer $k$, and measured in a complete basis $\{\ket{z}\}$.
For sufficiently late times $t$, the rescaled probabilities $\tilde{p}(z,t)$ follow the Porter-Thomas distribution, up to a correction bounded by the \textup{inverse effective Hilbert space dimension} $D_\beta^{-1}\equiv\sum_{z,E} \abs{\braket{z}{E}}^4 \abs{\braket{E}{\Psi_0}}^4/\pdiagz$, where $\{\ket{E}\}$ are the eigenstates of $H$.
\end{theorem*}
Rigorous statements of our Theorem and their proofs are presented in the SM~\cite{SM}. The only assumption of this Theorem is the $k$-th no-resonance condition, stating that $\sum_{i=1}^k E_{\alpha_i}$\,$=$\,$\sum_{i=1}^k E_{\beta_i}$ if and only if the $k$ indices $(\alpha_i)$ are a permutation of $(\beta_i)$~\cite{goldstein2006distribution,reimann2008foundation,linden2009quantum,kaneko2020characterizing,huang2021extensive}. That is, the eigenvalues $\{E_j\}$ of $H$ possess no resonant structures. This condition is expected to hold for generic ergodic Hamiltonians~\cite{goldstein2006distribution,reimann2008foundation}, and we find that it even holds in some integrable systems such as the 1D Fermi-Hubbard model [Fig.~\ref{fig:Fig1}(c)]~\cite{SM,essler2005one}.
The effective dimension $D_\beta$ quantifies the size of the Hilbert space explored during quench evolution (that can be probed in the $\{\ket{z}\}$ basis). {$D_\beta$ is similar to a participation ratio of $\ket{\Psi_0}$ and $\ket{z}$, when they are decomposed in the energy eigenbasis, which generically} increases exponentially with increasing system size, leading to a better agreement with the PT distribution and enabling our protocol to be increasingly accurate.

Our Theorem states that the outcome distribution factorizes into two parts $p(z,t) =p_\text{avg}(z) \times \tilde{p}(z,t)$: systematic values $p_\text{avg}(z)$ and random Porter-Thomas fluctuations $\tilde{p}(z,t)$. The systematic value is related to thermalization and does not distinguish between pure and mixed states, usually set by coarse-grained information such as the total energy {$\bra{\Psi_0}H\ket{\Psi_0}$}. Meanwhile, the fluctuations originate from random interference and average away in a mixed state~\cite{SM}.
These fluctuations are highly sensitive to details of the initial state and evolution, serving as a ``fingerprint" that enables their benchmarking. 

\textit{Speckle-based benchmarking.---}We provide an intuitive explanation of our benchmark $F_d$, based on two properties:
(i) the second moment of the rescaled $\tilde{p}(z)$ is $\mathcal{Z} = \sum_z p_\text{avg}(z) \tilde{p}(z)^2 \approx 2$ for $p(z)$ arising from ideal unitary evolution, (ii) the experimental distribution $q(z)$ in the presence of errors can be expressed as a linear combination $q(z)=F p(z) + (1-F) p_\perp (z)$, where $p_\perp (z)$ is uncorrelated with the ideal distribution $p(z)$ in the following sense:
$\mathbb{E}_z[ \tilde{p}(z)\tilde{p}_\perp(z)] \approx  \mathbb{E}_z[\tilde{p}(z)]\mathbb{E}_z[ \tilde{p}_\perp(z)] = 1$, with $\mathbb{E}_z [\cdot] \equiv \sum_z \pdiagz[\cdot]$ and $\tilde{p}_\perp \equiv p_\perp /p_\text{avg}$.
The second property relies on an assumption that the speckle pattern in $\tilde{p}$ significantly changes under any error, which has been rigorously proven for RUCs~\cite{dalzell2021random,gao2021limitations}.

Using these properties, our estimator $F_d$ is designed to isolate the desired ``fingerprint," taking value 1 when $q(z) = p(z)$ and 0 when $q(z) = p_\perp(z)$, which in turn implies $F_d\approx F$~\cite{SM}. We emphasize that it is essential to use the rescaled $\tilde{p}(z)$ to estimate the fidelity; otherwise $p(z)$ and $p_\perp(z)$ exhibit large correlations due to their shared systematic values $\pdiagz$.

In fact, the second condition can be relaxed. Under local coherent or incoherent errors, the relation $F_d \approx F$ can be shown at late times without any assumption on $q(z)$, solely based on the Eigenstate Thermalization Hypothesis and no-resonance conditions~\cite{SM}. We also argue that this result extends to multiple stochastic errors, and small coherent errors in the quench Hamiltonian, as verified by various numerical simulations~\cite{SM}. Furthermore, we also verify that $F_d \approx F$ holds even for relatively short quenches, well before the PT distribution emerges in $\tilde{p}(z)$, owing to the time-dependent adjustment factor $\mathcal{Z}(t)$ in $F_d$. This factor is inspired by $F_c$~\cite{choi2021emergent} and its effect is illustrated in~Fig.~\ref{fig:Fig1}(b-e) and the SM.

\textit{Measurement-basis independence.}---A surprising feature of our approach is that the fidelity is estimated from measurements in a fixed basis, despite the fact that the fidelity also depends on phase information not accessed from such measurements.
Nevertheless, our protocol works because quench evolution transforms the effects of physically relevant errors, including phase errors, into a form detectable by generic local measurements, after a short delay time [Fig.~\ref{fig:Fig3}(a)].

In our examples above, the quench dynamics plays two roles simultaneously: it enables our protocol, but also generates imperfect quantum evolution (due to errors), whose fidelity decay is measured.  
If the quench evolution were perfect, the measured fidelity would only reflect the state preparation error of the (potentially interesting) target initial state.
We present numerical demonstrations of such target state benchmarking in the SM~\cite{SM}.

\begin{figure}
    \centering
    \includegraphics[width=\columnwidth]{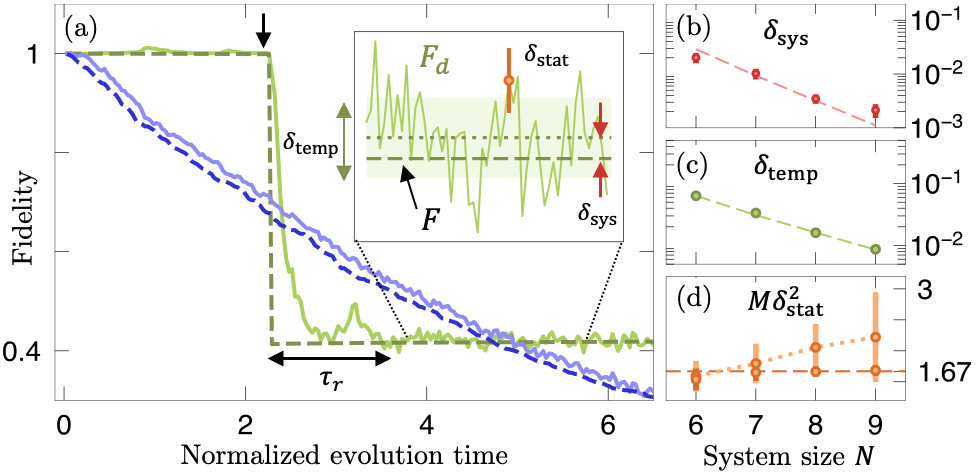}
    \vspace{-15pt}
    \caption{
    Performance analysis. (a)
    Numerical simulation of open system dynamics of a 1D Bose-Hubbard model with $N = 9$ particles on $N$ sites ($D= 24310$).  
    When a single error occurs (black arrow), $F_d$ (solid line, green) closely approximates the fidelity $F$ (dashed line, green) after a short delay time $\tau_r$.
    Averaging $F_d$ and $F$ over (potentially many) stochastic errors at different times gives their values for the mixed state $\rho$ (blue). The inset shows uncertainties and errors in our method. In particular, $F_d$ slightly deviates from $F$, quantified by the difference $\delta_\textrm{sys}$ between $F$ and the time-averaged $F_d$ (red arrow) and by the fluctuations $\delta_\textrm{temp}$ of $F_d$ over time (green arrow). Furthermore, a finite number of samples $M$ results in a statistical uncertainty $\delta_\textrm{stat}$ in the unbiased estimator $\hat{F}_d$ (orange error bar). (b,c) Both $\delta_\textrm{sys}$ and $\delta_\textrm{temp}$ decrease exponentially with system size, quantitatively agreeing with analytic predictions (dashed lines).
    (d) The sample complexity $M \delta^2_\textrm{stat}$ increases weakly with $N$ at fixed, early times (dotted line). At late times, it approaches the $N$-independent value $1+2F-F^2$ (dashed line).
    Error bars in (b-d) indicate variations over an ensemble of disordered Hamiltonians. See SM~\cite{SM} for details.}
    \label{fig:Fig3}
\end{figure}

\textit{Performance analysis.---}Our Theorem enables us to predict how well $F_d$ estimates the fidelity. First, we point out that it suffices to study the effect of a single error when the error rate is sufficiently small. This is because incoherent noisy dynamics can be ``unravelled" into an ensemble of stochastic pure state trajectories~\cite{daley2014quantum} each corresponding to a fixed occurrence of errors [Fig.~\ref{fig:Fig3}(a)]. As long as they are sufficiently infrequent, the effect of multiple errors can be understood from that of a single error~\cite{gao2021limitations,dalzell2021random}.

We showcase the performance of $F_d$ under realistic conditions by numerically simulating the 1D Bose-Hubbard model. See Figure~\ref{fig:Fig3}(a). We quantify the performance of $F_d$ along several axes~[Fig.~\ref{fig:Fig3}(a), inset]: the systematic error $\delta_\textrm{sys}$ refers to the difference between the true fidelity $F$ and the time-averaged $F_d$,
while the $\delta_\textrm{temp}$ quantifies how $F_d$ fluctuates over time. 
Finally, the statistical fluctuation (or, \textit{sample complexity}) $\delta_\textrm{stat}$ measures the uncertainty of the estimated $\hat{F}_d$ associated with a finite number of samples $M$, and hence the number of samples required to determine $F_d$ up to a desired precision. See SM~\cite{SM} for further details. 

Our Theorem allows analytical estimation of these quantities in the limit of long evolution: $\delta_\textrm{sys}$ and $\delta_\textrm{temp}$ are respectively $O(D_\beta^{-1})$ and $O(D_\beta^{-1/2})$~\cite{SM}. Hence, both the accuracy and precision of our benchmark improve exponentially with increasing system size, explicitly confirmed in our numerical simulations~[Fig.~\ref{fig:Fig3}(b,c)]. Meanwhile, the sample complexity has optimal scaling. It is system size independent for long evolutions: $M\delta^2_\textrm{stat}\approx 1+2F-F^2$~[Fig.~\ref{fig:Fig3}(d), dashed line]. For short quench evolution, the sample complexity grows weakly with system size [Fig.~\ref{fig:Fig3}(d), dotted line]. Finally, in the presence of incoherent errors, the finite response time $\tau_r$ leads to a slight delay between $F_d$ and the continuously-decaying fidelity~\cite{gao2021limitations,mi2021information,SM}.

\begin{figure}[tpb]
    \centering
    \includegraphics[width=\columnwidth]{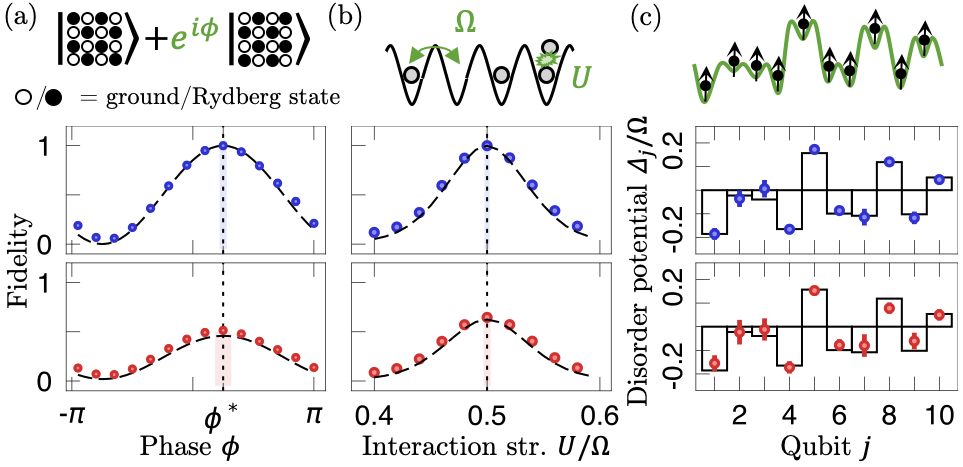}
    \vspace{-15pt}
    \caption{
    State and process benchmarking with $F_d$.
    We numerically simulate the estimation of: (a) the phase $\phi$ of a GHZ-like initial state in a 2D Rydberg system; (b) the ratio of the interaction and tunneling strengths $U/\Omega$ in a 1D Bose-Hubbard model; and (c) ten disordered on-site potentials in a trapped ion model.
    Parameters are estimated by maximizing $\hat{F}_d$ over simulated parameter values, using measurements after error-free (blue) or noisy (red) quench evolution.
    The error bars and shaded regions indicate the statistical uncertainties in $\hat{F}_d$ and the parameter values with 1000 samples. See SM~\cite{SM} for details.
    (a,b) Both $F$ (black lines) and $\hat{F}_d$ (markers) are consistent and simultaneously maximized at the true parameter value (dotted line).
    (c) Reconstructed disorder potential values (markers) are consistent with their true values (lines).
    }
    \label{fig:Fig4}
\end{figure}

\textit{Limitations.---}While our protocol is applicable to generic quantum many-body systems, it may fail in special cases in which our Theorem is not applicable.
Examples include systems with weakly- or non-ergodic dynamics, or the presence of correlated nonlocal errors. We provide detailed analysis and potential resolutions of known failure cases in the SM~\cite{SM}.

\textit{Applications.---}
The ability to measure the fidelity enables further applications. As examples, we show that one can simultaneously estimate multiple parameters of quantum states or Hamiltonians. The key observation is that, given the ability to measure the fidelity between a theoretical model and experiment, one can vary classical simulation model parameters to maximize the estimated fidelity~\cite{arute2019quantum,choi2021emergent}.
This optimization requires classical computation but no further data acquisition.
We numerically demonstrate this idea in three different examples and verify that the extracted parameter values are close to the actual ones, even in the presence of noise. See Figure~\ref{fig:Fig4}.

{\it Acknowledgments.---} 
We would like to thank Wen~Wei Ho, Anant Kale, Eun~Jong Kim, Sherry Zhang, and Peter Zoller for useful discussions. 
In particular, we thank Sherry Zhang for a careful reading of this manuscript and Anant Kale for insightful discussion at the early stage of the project. We acknowledge funding provided by the NSF Physics Frontiers Centers IQIM (NSF PHY-1733907) and CUA (NSF PHY-1734011), the AFOSR YIP (FA9550-19-1-0044), the DARPA ONISQ program (W911NF2010021), the Army Research Office MURI program (W911NF2010136), the NSF QLCI program (2016245), the DOE (DE-SC0021951), NSF CAREER award 1753386 and NSF CAREER award 2237244. Support is also acknowledged from the U.S. Department of Energy, Office of Science, National Quantum Information Science Research Centers, Quantum Systems Accelerator. JC acknowledges support from the IQIM postdoctoral fellowship. ALS acknowledges support from the Eddleman Quantum graduate fellowship.

\let\oldaddcontentsline\addcontentsline% Store \addcontentsline
\renewcommand{\addcontentsline}[3]{}% Make \addcontentsline a no-op
\let\addcontentsline\oldaddcontentsline% Restore \addcontentsline
\renewcommand{\theequation}{S\arabic{equation}}
\renewcommand{\thefigure}{S\arabic{figure}}
\setcounter{equation}{0}
\setcounter{figure}{0}

\newpage
\begin{widetext}
\appendix
\clearpage
\maketitle
\vspace{-20pt}
{
    \center \bf \large 
    Supplemental Material for: \\
Benchmarking quantum simulators using ergodic quantum dynamics\vspace*{0.1cm}\\ 
    \vspace*{0.0cm}
}
\tableofcontents
{\section{Previous fidelity benchmarks}
The cross-entropy benchmark was introduced in Ref.~\cite{boixo2018characterizing} and was prominently used in Ref.~\cite{arute2019quantum} to demonstrate quantum computational advantage. It is defined as
\begin{equation}
    F_\text{XEB} \equiv D \sum_{z} p(z) q(z) - 1~,
\end{equation}
where as in the main text, $p(z)$ and $q(z)$ are the probabilities of obtaining the outcome $z$ from the ideal simulated state and experimental states respectively, and $D$ is the Hilbert space dimension, here $D=2^N$ where $N$ is the number of qubits. $F_\text{XEB}$ has been argued to estimate the fidelity in deep random unitary circuits (RUCs)~\cite{boixo2018characterizing,arute2019quantum,dalzell2021random,gao2021limitations}.

The estimator $F_c$ was introduced in Ref.~\cite{choi2021emergent}, defined as
\begin{equation}
    F_c \equiv 2 \frac{\sum_z p(z) q(z)}{\sum_z p(z)^2}-1~.
\end{equation}
$F_c$ is a slight modification of $F_\text{XEB}$ to reliably estimate the fidelity in shallow circuits. This is because of the onset of ``local scrambling" even in shallow circuits, discussed in Ref.~\cite{cotler2021emergent}. In Ref.~\cite{choi2021emergent}, $F_c$ was experimentally demonstrated to estimate the fidelity in a Rydberg atom array quantum simulator, at infinite effective temperature and short quench evolution times. $F_d$ can be taken as a further improvement of $F_c$, that works for arbitrary quench initial states and Hamiltonians.}

\section{Emergence of universal statistical properties from Hamiltonian dynamics}
\label{app:Hamil_k_design}
In this section, we present precise statements of our Theorem in the main text and prove them.
The significance of our Theorems~\ref{thm:temporal_PT} and~\ref{thm:bitstring_PT} is that a Porter-Thomas distribution emerges from measurement outcomes $p(z,t) = \abs{\braket{z}{\Psi(t)}}^2$ after generic Hamiltonian dynamics, upon appropriate data processing.
To prove our theorems, we first establish a new identity, the \textit{$k$-th Hamiltonian twirling identity}, which only requires one assumption: the eigenvalues of the Hamiltonian $H$ satisfy the \textit{$k$-th no-resonance condition}. We utilize this identity to establish theorems on the emergence of a Porter-Thomas distribution in the rescaled measurement outcomes $\tilde{p}(z,t)$. 
\subsection{$k$-th no-resonance condition}
\noindent \textbf{Definition ($k$-th no-resonance):} Let $E_j$ be the energy eigenvalues of a Hamiltonian $H$.
$H$ satisfies the $k$-th \textit{no-resonance condition}~\cite{reimann2008foundation,linden2009quantum,kaneko2020characterizing,huang2021extensive} if
\setcounter{equation}{2}
\begin{equation}
    E_{\alpha_1} + E_{\alpha_2} + \cdots + E_{\alpha_k} = E_{\beta_1} + E_{\beta_2} + \cdots + E_{\beta_k}~ \label{eq:k_no_resonance}
\end{equation}
implies that the indices $(\alpha_1, \alpha_2, \dots, \alpha_k)$ and  $(\beta_1, \beta_2, \dots, \beta_k)$ are equal up to reordering.
In other words, $\beta_j = \alpha_{\sigma(j)}$ for some permutation $\sigma \in S_k$, where $S_k$ is the group of permutations of $k$ elements (the symmetric group). It is easy to see that if the $k$-th no-resonance condition is satisfied, then so are the $k'$-th no-resonance conditions, for $k'<k$.
This assumption is considered generically satisfied: it is believed that if the condition does not hold in a given ergodic system, any small perturbation will generically break any $k$-th resonances~\cite{kaneko2020characterizing}.

\subsection{$k$-th no resonance with degeneracies}
The $k$-th no-resonance conditions are, strictly speaking, violated if the Hamiltonian contains degeneracies: $E_i = E_j, i \neq j$. This violates the first no-resonance condition and so all higher $k$-th no-resonance conditions. This seems to preclude systems with symmetry-protected degeneracies, which can still be ergodic. For our purposes, however, we will consider \textit{pure states} $\rho = \ketbra{\Psi_0}{\Psi_0}$. In this case, one can replace any set of degenerate states $\mathcal{E} = \{\ket{\Psi_E}\}$ with a \emph{single} representative state, defined as the projection of $\ket{\Psi_0}$ onto the degenerate subspace: $\mathcal{P}_\mathcal{E}\ket{\Psi_0}$. 
In this way, all degeneracies are effectively eliminated.
Therefore, while Lemma~\ref{lemma:Hamiltonian_k_design} (below) requires the full $k$-th no-resonance condition, one can modify this condition by demanding the $k$-th no-resonance condition only after eliminating all degenerates as prescribed above. 
Our main Theorems only require this weakened version of the $k$-th no-resonance condition.
Hence, our Theorems holds generically even in the presence of degeneracies protected by symmetries.

The $k$-th no-resonance condition is also notably violated for $k\geq 2$ if there is a \textit{particle-hole symmetry}, in which the energy spectrum is symmetric about $E=0$. This case is discussed in Section~\ref{app:PHS}.

{We finally remark that there are other, indirect ways in which degeneracies can affect the performance of our benchmark. For example, if there is a large degeneracy of eigenstates, the quench dynamics will be less ergodic and our protocol will have larger uncertainties. However, this is entirely captured by our effective dimension parameter $D_\beta$, and no further modification is required.}

\subsection{$k$-th Hamiltonian twirling identity}
We first present the $k$-th Hamiltonian twirling identity. This identity is a rewriting of a $k$-fold quantum channel: we take a state $\rho^{(k)}$ in $k$-copies of the Hilbert space $\mathcal{H}^{\otimes k}$, and average it over all possible times of time-evolution $U_t^{\otimes k}$ over $k$ copies, with $U_t \equiv \exp(-iHt)$. That is, the channel $C^{(k)}$ is $C^{(k)}[\rho^{(k)}] = \Ebb_t\L[U_t^{\otimes k} \rho^{(k)} U_{-t}^{\otimes k}\R]$ where $\Ebb_t[f(t)] \equiv \lim_{T\rightarrow\infty} \frac{1}{T} \int_0^T f(t) dt$ here and below.
We dub this the $k$-th Hamiltonian twirling channel, in analogy to twirling channels, in which a state is averaged over all possible unitaries in a group, most commonly the ensemble of Haar unitaries or the Clifford or Pauli groups~\cite{dankert2009exact}.

We find that the $k$-th Hamiltonian twirling channel can be written as sum of (i) a dephasing operator in the energy eigenbasis, followed by a sum of permutations of $k$ elements and (ii) additional terms associated with the double counting of certain permutations. We associate these terms with additional correlations associated with energy conservation, hence we dub these ``energy-conservation terms".

Such a rewriting is motivated by \textit{Haar-unitary twirling channels}, which can be written as a sum of permutations~\cite{gross2007evenly}. In our Theorem below, we consider a specific quantity: $\Ebb_t[p^k(z,t)] = \bra{z}^{\otimes k} C^{(k)}[\ketbra{\Psi_0}{\Psi_0}^{\otimes k}] \ket{z}^{\otimes k}$. In this case, the terms in (i) are the dominant contributions and we are able to bound the magnitude of the energy-conservation terms in (ii).

We explicitly state our identity for $k=2$.
\begin{lemma*}[Second Hamiltonian twirling identity] If a Hamiltonian $H$ satisfies the second no-resonance condition, the following identity holds:
\begin{align}
    &C^{(2)}\big[\rho^{(2)} \big] \equiv \Ebb_t \L[ U_t^{\otimes 2} \rho^{(2)} U_{-t}^{\otimes 2}\R] \label{eq:Ham_2_design}\\
    &= \sum_{E,E'} \bigg[ \ketbra{E,E'}{E,E'} \rho^{(2)} \ketbra{E,E'}{E,E'} + \ketbra{E,E'}{E,E'} \rho^{(2)} \ketbra{E',E}{E',E} \bigg] - \sum_E \ketbra{E,E}{E,E} \rho^{(2)} \ketbra{E,E}{E,E} \nonumber\\
    &=  \mathcal{D}^{\otimes 2}\big[\rho^{(2)}\big]  \mathbb{I} + \mathcal{D}^{\otimes 2}\big[\rho^{(2)} \mathcal{S}\big] \mathcal{S} - \sum_E \ketbra{E,E}{E,E} \rho^{(2)} \ketbra{E,E}{E,E}~, \nonumber
\end{align}
where $U_t \equiv \exp(-iHt)$, $\mathbb{E}_t[f(t)] = \lim_{T\rightarrow \infty} \frac{1}{T} \int_0^T f(t) dt$ is the average of $f(t)$ over infinite time. $\rho^{(2)}$ is any state on two copies of the Hilbert space, and $\ket{E}$ are energy eigenstates of the Hamiltonian. $\mathbb{I}$ and $\mathcal{S}$ are the identity and swap operators on $\mathcal{H}^{\otimes 2}$, defined as $\mathbb{I}(\ket{\psi_1} \otimes \ket{\psi_2}) = \ket{\psi_1} \otimes \ket{\psi_2}$ and $\mathcal{S}(\ket{\psi_1} \otimes \ket{\psi_2}) = \ket{\psi_2} \otimes \ket{\psi_1}$ for any $\ket{\psi_1},\ket{\psi_2}$. Lastly, $\mathcal{D}^{\otimes 2}$ is the dephasing channel in the energy eigenbasis: $\mathcal{D}^{\otimes 2 }\big[\rho^{(2)}\big] \equiv \sum_{E,E'} \ketbra{E,E'}{E,E'}\rho^{(2)}\ketbra{E,E'}{E,E'}$. Finally, we remark that the state $\rho^{(2)}$ on $\mathcal{H}^{\otimes 2}$ can be generalized to any operator $A$ acting on $\mathcal{H}^{\otimes 2}$.
\end{lemma*}

\begin{proof}
\begin{align}
    &\mathbb{E}_t \L[U_t^{\otimes 2} \rho^{(2)}  U_{-t}^{\otimes 2} \R] =  \lim_{T\rightarrow \infty}\frac{1}{T}\int^T_0 dt~U_t^{\otimes 2} \rho^{(2)}  U_{-t}^{\otimes 2} \\
    &=  \sum_{E_1,E_2,E_3,E_4} \L[\lim_{T\rightarrow \infty}\frac{1}{T}\int^T_0 dt~e^{-i(E_1-E_2+E_3-E_4)t}\R] \ketbra{E_1,E_3}{E_1,E_3}\rho^{(2)} \ketbra{E_2,E_4}{E_2,E_4} \\
    &=  \sum_{E_1,E_2,E_3,E_4} \delta(E_1-E_2+E_3-E_4)~ \ketbra{E_1,E_3}{E_1,E_3}\rho^{(2)} \ketbra{E_2,E_4}{E_2,E_4}\\
    &\overset{\text{$2$-NR}}{=}  \sum_{E,E'} \bigg[ \ketbra{E,E'}{E,E'} \rho^{(2)}  \ketbra{E,E'}{E,E'} + \ketbra{E,E'}{E,E'} \rho^{(2)}  \ketbra{E',E}{E',E} \bigg] - \sum_E \ketbra{E,E}{E,E} \rho^{(2)}  \ketbra{E,E}{E,E}~, \label{eq:lemma_2}
\end{align}
where ``(2-NR)" indicates the use of the second no-resonance condition. The second no-resonance condition states that the delta function $\delta(E_1-E_2+E_3-E_4)$ is non-zero in only two cases: $E_1 = E_2, E_3=E_4$ or $E_1 = E_4, E_3=E_2$, which give the two terms above. The final, energy-conservation term in Eq.~\eqref{eq:lemma_2} corresponds to $E_1=E_2=E_3=E_4$. 
This special case belongs to the above two terms simultaneously and is double counted.
Therefore, the energy-conservation term is subtracted to correct for this double counting.\end{proof}
A diagrammatic representation of this Lemma is shown in  Fig.~\ref{fig:2design}(a,b). 
The energy-conservation terms are illustrated in Fig.~\ref{fig:2design}(c).
As with the second Hamiltonian twirling identity, one can define the $k$-th twirling identity.
\begin{lemma}[$k$-th Hamiltonian twirling identity] \label{lemma:Hamiltonian_k_design}
If the Hamiltonian satisfies the $k$-th no-resonance condition, its $k$-th twirling channel is:
\begin{align}
    C^{(k)}[\rho^{(k)}] &\equiv \mathbb{E}_t \L[ U_t^{\otimes k} \rho^{(k)} U_{-t}^{\otimes k}\R] \nonumber\\
    &= \sum_{M = (E_1,\cdots, E_k)} \bigg[\sum_{\sigma \in \pi(M)} \ketbra{E_1, \cdots, E_k}{E_1, \cdots, E_k} \rho^{(k)}  \ketbra{E_{\sigma(1)}, \cdots, E_{\sigma(k)}}{E_{\sigma(1)}, \cdots, E_{\sigma(k)}} \bigg],  \label{eq:multiset_line}\\
    &= \mathcal{D}^{\otimes k}\big[\rho^{(k)}\hat{\sigma}^{-1} \big]\sum_{\sigma \in S_k} \hat{\sigma} + \mathrm{E.C.} \label{eq:Ham_k_design}
\end{align}
where $\mathbb{E}_t[f(t)] = \lim_{T\rightarrow \infty} \frac{1}{T} \int_0^T f(t) dt$ is averaging over infinite time, $\rho^{(k)}$ is any state on $k$ copies of the Hilbert space $\mathcal{H}^{\otimes k}$, and $\hat{\sigma}$---with $\sigma \in S_k$---is a permutation operator acting on the $k$-copy Hilbert space: 
\begin{align}
    \hat{\sigma} = \sum_{\{E_i\}} \ketbra{E_{1},E_{2}, \dots, E_{k}}{E_{\sigma(1)},E_{\sigma(2)}, \dots, E_{\sigma(k)}}.
\end{align} 
In~\eqref{eq:multiset_line}, the sum is taken over all possible lists of $k$ energies $M=(E_1,\cdots,E_k)$. If $M$ has $r$ unique energies with multiplicities $n_1,\cdots,n_r$ (such that $\sum_{i=1}^n n_r = k$), its group of permutations $\pi(M)$ has $k!/\prod_{i=1}^n n_i!$ elements and is a subgroup of the symmetric group $S_k$. In~\eqref{eq:Ham_k_design}, this identity is re-written in terms of permutations, the dephasing channel $\mathcal{D}^{\otimes k}\big[\rho^{(k)}\big] \equiv \sum_{E_1,\cdots, E_k} \ketbra{E_1,\cdots, E_k}{E_1,\cdots, E_k} \rho^{(k)} \ketbra{E_1,\cdots, E_k}{E_1,\cdots, E_k}$, and energy-conservation terms $\text{E.C.}$ associated with double, or multiple, counting (similar to the $k=2$ case). 
All terms can be explicitly derived from \eqref{eq:multiset_line}.
\end{lemma}

The proof of the $k$-th twirling identity is analogous to that of the second twirling identity: in this case, there are $k!$ permutations relevant for the $k$-th no resonance theorem, along with more energy-conservation terms than the simplest one in the $k=2$ case. 
As an example, the third twirling identity is illustrated diagrammatically in Fig.~\ref{fig:2design}(d). In addition to the energy-conservation terms associated with cases where two energies are the same $E_i = E_j$, there is a higher order energy-conservation term when all three energies are the same: $E_1 = E_2 = E_3$.

Lastly, we note that this identity also holds in Floquet systems, as long the $k$-th no resonance condition is tightened to $\sum_i E_{\alpha_i} \equiv \sum_i E_{\beta_i} (\text{mod }2\pi) \iff \{\alpha_i\} = \{\beta_i\}$, and the integral over all time is replaced by a discrete sum over integer Floquet periods $nT$.

\begin{figure}[tb]
    \centering
    \includegraphics[width=0.95\textwidth]{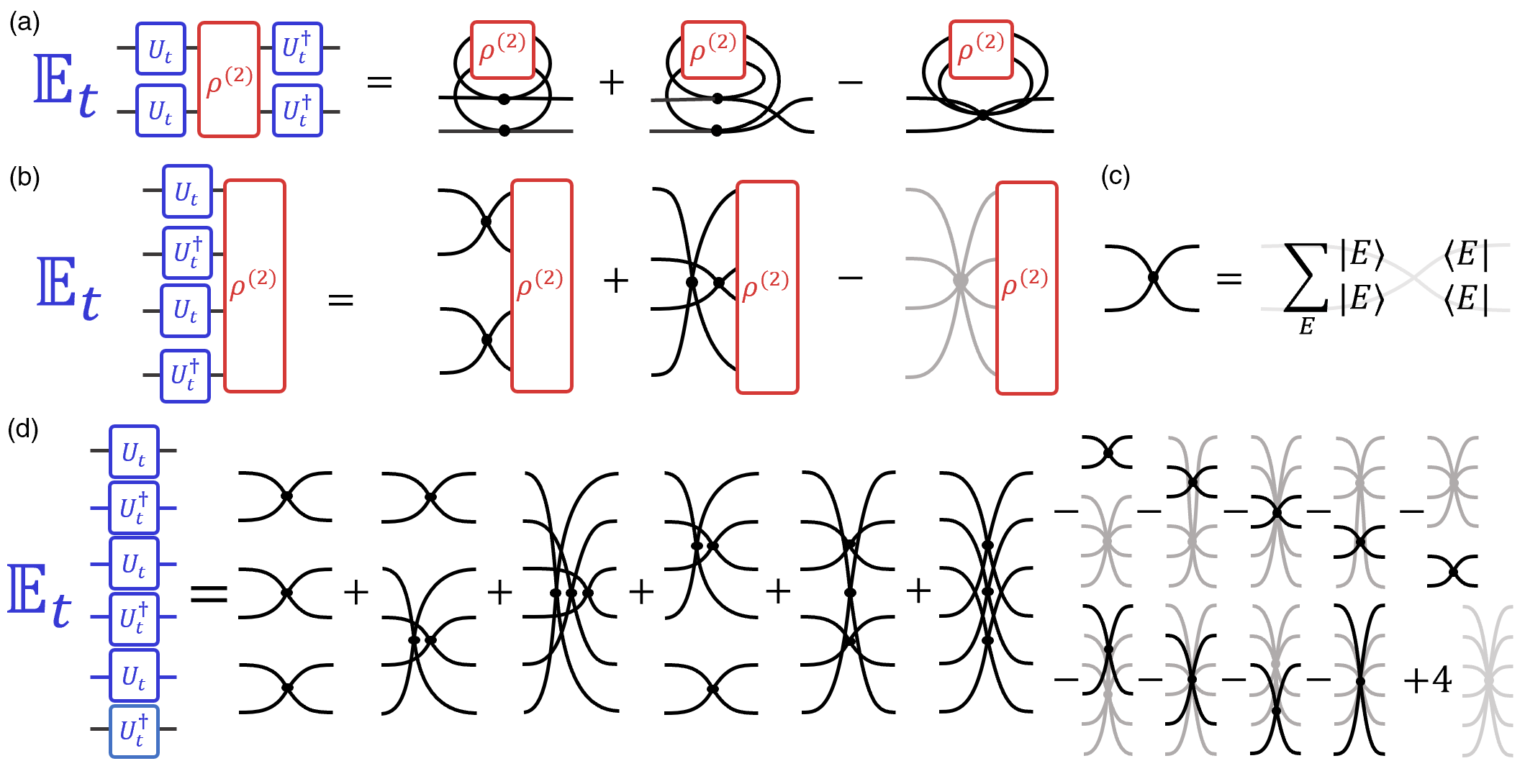}
    \caption{Second and third Hamiltonian twirling identities.
    (a) The second Hamiltonian twirling identity comprises identity and swap permutations terms, along with an \textit{energy-conservation} correction term that in our context will have an exponentially small contribution.
    (b) The same as (a), but open legs representing the channel input/outputs are rearranged in order to simplify the diagram. We keep this convention for the rest of diagrams in this document. 
    (c) Diagrammatic notation for the basic building block in our diagrams: the dephasing operator in the energy eigenbasis: $D[\rho] = \sum_E \ketbra{E}{E}\rho\ketbra{E}{E}$.
    (d) The third Hamiltonian twirling identity comprises of six permutation and several energy-conservation terms. For our purposes, these additional terms will have small contributions.}
    \label{fig:2design}
\end{figure}

\begin{figure}[t]
    \centering
    \includegraphics[width=0.95\textwidth]{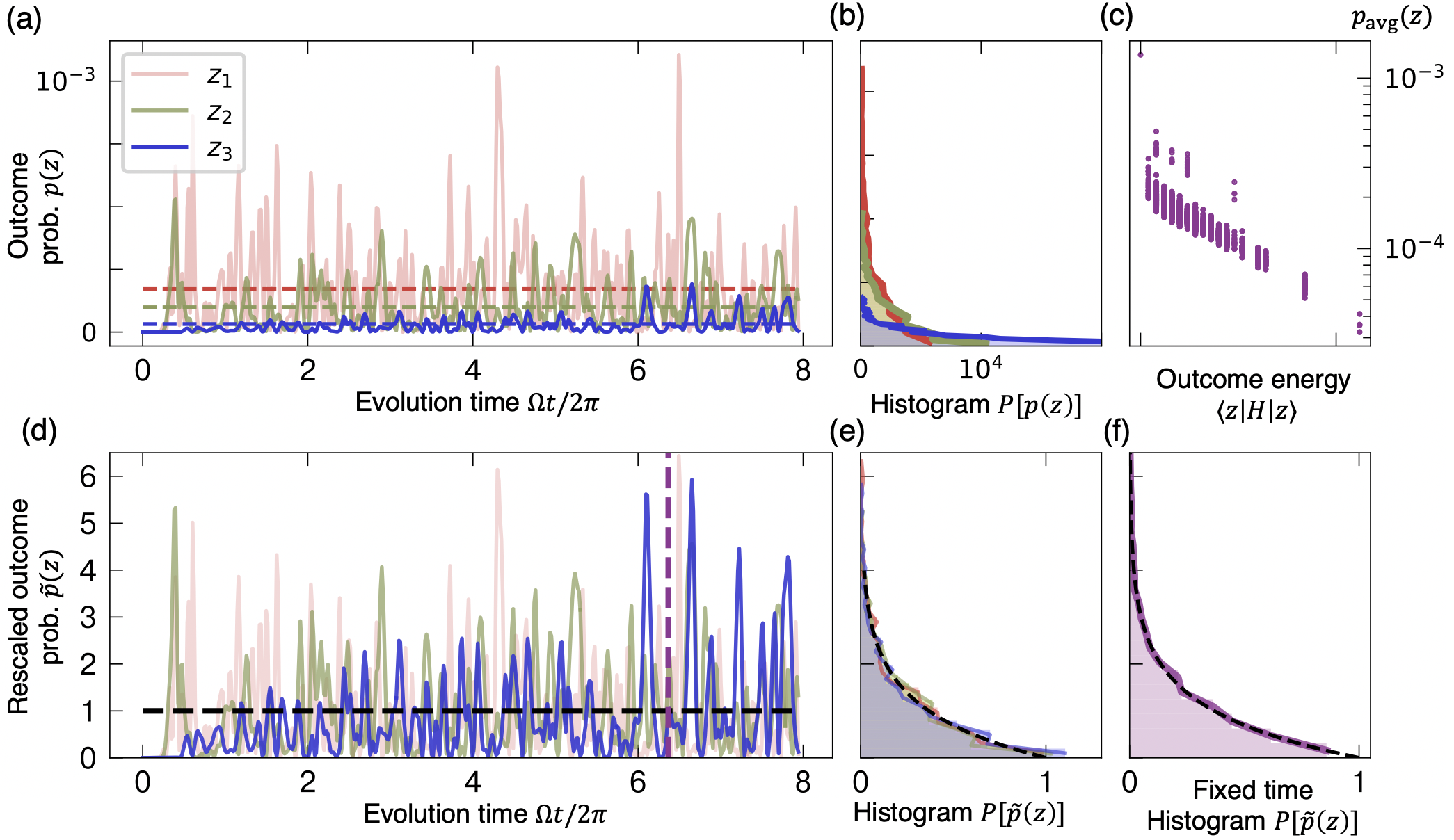}
    \caption{(a) Outcome probabilities $p(z,t)$ over time, for three outcomes $z_1,z_2,z_3$ in a Bose-Hubbard model with 8 particles on 8 sites in red, green and blue respectively. For each $z$, $p(z,t)$ fluctuates about its mean value $\pdiagz$ (dashed lines). (b) Plotting the histogram of $p(z,t)$ over time, each histogram follows an exponential distribution with differing means $\pdiagz$. (c) $\pdiagz$ is approximately exponentially dependent on the outcome energy $\bra{z}H\ket{z}$, analogous to a Boltzmann distribution. (d) The rescaled outcome probabilities $\tilde{p}(z,t)$ follow (e) a universal exponential (Porter Thomas) distribution, which follows from Theorem~\ref{thm:temporal_PT}. (f) Taking a histogram at fixed time [purple, (d)] over bitstrings $z$ gives the same distribution (Theorem~\ref{thm:bitstring_PT}).}
    \label{fig:theorem_graphic}
\end{figure}

\subsection{Porter-Thomas distribution over time}
We will use the $k$-th Hamiltonian twirling identity, with $\rho^{\otimes k} = \ketbra{\Psi_0}{\Psi_0}^{\otimes k}$, to prove that the quantity $\tilde{p}(z,t)\equiv p(z,t)/\pdiagz$ asymptotically follows the \textit{Porter-Thomas}, or exponential distribution~\cite{porter1956fluctuations}.

The Porter-Thomas (PT) distribution naturally arises from the distribution of overlaps $p$ of the Haar ensemble --- the unique distribution of states on the Hilbert space that is invariant under any unitary --- to a fixed vector $\ket{\Phi}$: $p = \abs{\langle \Psi | \Phi \rangle}^2$, with $\ket{\Psi}$ drawn from the Haar ensemble. A given $\ket{\Psi}$ and $\ket{\Phi}$ have overlap $\abs{\langle \Psi | \Phi \rangle}^2$. The probability density of obtaining an overlap with value $p$ is $P(p) = \int d\Psi \delta(\abs{\langle \Psi | \Phi \rangle}^2 - p)$. In deep random circuits, the wavefunction $\ket{\Psi}$ is approximately Haar-random and we can evaluate $P(p)$ for any Hilbert space dimension $D$, which converges to the Porter-Thomas distribution in the limit $D\rightarrow \infty$.
\begin{align}
    P(p) dp = (D-1)(1-p)^{D-2} dp \xrightarrow{D\rightarrow\infty} \exp(-D p) d(Dp)~,
    \label{eq:PorterThomas}
\end{align}
It is natural to rescale the quantity $p$ by $D$. In this case, the exponential distribution has a $k$-th moment $\int_0^\infty (Dp)^k P(D p) d(Dp) = k!$. 

We formally state the Theorem from the main text in Theorems~\ref{thm:temporal_PT} and \ref{thm:bitstring_PT} below. Informally, our theorems show that the rescaled distribution $\tilde{p}(z,t) \equiv p(z,t)/\pdiagz$ asymptotically follows the PT distribution. As a reminder, $p(z,t) = \abs{\braket{z}{\Psi(t)}}^2$ is the probability of measuring outcome $z$, and $\pdiagz = \Ebb_t\L[p(z,t)\R]$ is its time-averaged value. By showing that the $k$-th moments of $\tilde{p}(z,t)$ are close to $k!$, we show that $\tilde{p}(z,t)$ asymptotically approaches the Porter-Thomas distribution. We can choose to fix the bitstring $z$, and consider the ensemble over times $t$ [Theorem~\ref{thm:temporal_PT}, Fig.~\ref{fig:theorem_graphic}(d,e)], or fix a large but arbitrary time $t$, and consider the ensemble of $\tilde{p}(z,t)$ over outcomes $z$ [Theorem~\ref{thm:bitstring_PT}, Fig.~\ref{fig:theorem_graphic}(d,f)]. For brevity, we will occasionally denote these ensembles as $\{\tilde{p}(z,t)\}_t$ and $\{\tilde{p}(z,t)\}_z$ respectively, with the subscript denoting the variable over which the ensemble is formed (the other variable is kept fixed). Formally stated, our first Theorem is the following:

\begin{theorem}[Temporal Porter-Thomas distribution] \label{thm:temporal_PT}
For an initial state $\ket{\Psi_0}$ and a Hamiltonian $H$ satisfying the $k$-th \textit{no-resonance condition}, the temporal $k$-th moment of $\tilde{p}(z,t)$ (with fixed $z$), approaches a universal value $\Ebb_t\L[\tilde{p}(z,t)^k\R] = k! +  O(D_\beta^{-1}(z))$. $D_\beta(z)$ is the effective Hilbert space dimension of the outcome $z$ defined as $D_\beta^{-1}(z) \equiv \L(\sum_{E} \abs{\braket{z}{E}}^4 \abs{\braket{E}{\Psi_0}}^4\R)/\pdiagz^2$.
\end{theorem}

\begin{proof}
Our proof is a direct application of our $k$-th Hamiltonian twirling identity on the temporal $k$-th moment $\Ebb_t \L[\tilde{p}(z,t)^k\R] = k! + O(D_\beta^{-1})$:
To compute $\Ebb_t\L[\tilde{p}(z,t)^k\R]$, we apply Lemma~\ref{lemma:Hamiltonian_k_design} on $\rho^{(k)} = \ketbra{\Psi_0}{\Psi_0}^{\otimes k}$:
\begin{align}
    &\mathbb{E}_t\L[ \tilde{p}(z,t)^k\R] =  \frac{\bra{z}^{\otimes k} \mathbb{E}_t\L[U_t^{\otimes k} \ketbra{\Psi_0}{\Psi_0}^{\otimes k} U_{-t}^{\otimes k} \R] \ket{z}^{\otimes k}}{\pdiagz^{k}} \\
 &= \frac{\sum_{M = (E_1,\cdots, E_k)} \sum_{\sigma\in \pi(M)} ~\L[\prod_{i=1}^k \abs{\braket{z}{E_i}\braket{E_{\sigma_i}}{\Psi_0}}^2\R]}{\pdiagz^k}  = \frac{\sum_{M = (E_1,\cdots, E_k)} \abs{\pi(M)}~\L[\prod_{i=1}^k \abs{\braket{z}{E_i}\braket{E_i}{\Psi_0}}^2\R]}{\pdiagz^k}
\end{align}
where we have used Eq.~\eqref{eq:multiset_line}.
Now using Eq.~\eqref{eq:Ham_k_design}, where the contribution from the first term is equivalent to replacing $\abs{\pi(M)}$ with $k!$. This gives:
\begin{align}
    \frac{
    \bra{z}^{\otimes k} 
    \mathcal{D}^{\otimes k}
    \L[ \ket{\Psi_0}\bra{\Psi_0}^{\otimes k}
    \R]
    \L[\sum_{\sigma \in S_k} \hat{\sigma}\R]
    \ket{z}^{\otimes k}
    }
    {\pdiagz^k}
    =k! 
    \frac{\sum_{\{E_j\}}
    \prod_{j=1}^k \abs{\braket{z}{E_j}\!\braket{E_j}{\Psi_0}}^2  }
    {\pdiagz^k}
    = k!~,
\end{align}
where in the first equality we used the fact that $\hat{\sigma}$ operator acts trivially on the product state $\ket{z}^{\otimes k}$ and in the second equality, we used the relation 
\begin{equation}
\pdiagz \equiv \mathbb{E}_t \L[ p(z,t)\R] = \sum_E \abs{\braket{z}{E}\!\braket{E}{\Psi_0}}^2~.
\end{equation}
This relation is used several times, motivating our definition of the quantity 
\begin{equation}
    q_z(E) \equiv \frac{\abs{\braket{z}{E}\!\braket{E}{\Psi_0}}^2}{\pdiagz}
\end{equation}
In fact, $q_z(E)$ has the properties of a probability distribution over the energies $E$: $0 \leq q_z(E) \leq 1$ and $\sum_E q_z(E) = 1$.

Now we wish to bound the deviation $\abs{\mathbb{E}_t\L[ \tilde{p}(z,t)^k\R]-k!}$, that is the contribution from all energy-conservation (E.C.) terms. To do so, we note the E.C. terms must come from lists of energies $M = (E_1,\cdots,E_k)$ with at least one repeated energy. That is:
\begin{align}
    \abs{\mathbb{E}_t\L[ \tilde{p}(z,t)^k\R]-k!} \equiv \abs{\frac{\bra{z}^{\otimes k}(\text{E.C.})\ket{z}^{\otimes k}}{\pdiagz^k}} &= \sum_{\substack{M=(E_1,\cdots,E_k)\\\text{not all $E_j$ distinct}}} \abs{\pi(M)-k!} \prod_{j=1}^k q_z(E_j)\\
    &\leq k!\L[ \sum_{E_1=E_2}  \sum_{E_3,\cdots, E_k}\prod_{j=1}^k q_z(E_j) + \sum_{E_1=E_3}  \sum_{E_2,E_4,\cdots, E_k}\prod_{j=1}^k q_z(E_j) + \cdots \R] \nonumber
\end{align}

In the last line, we have used two separate inequalities: we bound the difference $\abs{\pi(M)-k!}\leq k!$. This enables us to contain all possible lists $M$ with at least one repeated energy in the sums above, where we have only explicitly demanded that one pair of energies coincide: $E_a = E_b$. There are ${k \choose 2}$ choices of $a$ and $b$. Each choice has the same value. Therefore without loss of generality, we set $a=1$ and $b=2$. Then we obtain:
\begin{align}
    \abs{\mathbb{E}_t\L[ \tilde{p}(z,t)^k\R]-k!} \leq k! {k \choose 2} \sum_{E_1} q_z(E_1)^2  \prod_{j=3}^k \sum_{E_j} q_z(E_j) =  k! {k \choose 2} \sum_{E_1} q_z(E_1)^2 \equiv k! {k \choose 2} D_\beta^{-1}(z)~.
\end{align}
Therefore, we obtain the desired bound
\begin{equation}
    \abs{\mathbb{E}_t\L[ \tilde{p}(z,t)^k\R]-k!} \leq k! {k \choose 2} D_\beta^{-1}(z) = O(D_\beta^{-1}(z))~. \label{eq:}
\end{equation}
\end{proof}
In fact, we expect this bound to be close to tight: the contribution from the simplest energy-conservation terms --- those with exactly one repeated energy --- is close to our bound. Other energy-conservation terms are expected to be exponentially smaller than these simplest terms. Here, the factor $\pi(M) = k!/2$, and we obtain:
\begin{equation}
    \mathbb{E}_t\L[\tilde{p}(z,t)^k\R] \approx k! \L( 1- \frac{k(k-1)}{2} \frac{D_\beta^{-1}(z)}{2} \R)~.
    \label{eq:PT_subleading}
\end{equation}
Eq.~\eqref{eq:PT_subleading} agrees with the leading order (in $D^{-1}$) correction to the $k$-th moment of the Porter-Thomas distribution:
\begin{align}
   \int_0^1 dp~p^k (D-1)(1-p)^{D-2} = k! \frac{D^k}{D \cdots (D+k-1)} =k! \L[ 1- \frac{k(k-1)}{2} D^{-1} + O(D^{-2}) \R]~. \label{eq:PT_kth_moment_expand}
\end{align}
Therefore, it is apt to call  $D_\beta(z)$ as the effective dimension: the deviation of the $k$-th moment from $k!$ for the temporal distribution is the same as that of the Porter-Thomas distribution for a system with finite dimension $2D_\beta$. $\{\tilde{p}(z,t)\}_t$ approximately follows a Porter-Thomas distribution over a vector space of dimension $2D_\beta(z)$. Note, however, that this agreement does not extend to corrections of higher orders in $D^{-1}$.

For a fixed $z$, the inverse effective dimension $D_\beta^{-1}(z)$ can be understood as the \textit{collision probability} associated with $q_z(E)$.
More explicitly, we can define $D_\beta^{-1}(z) = \sum_E q_z(E)^2$ such that $D_\beta (z)$ describes the participation ratio (or equivalently the size of effective Hilbert space dimension) corresponding to a particular bitstring $z$.
Then the overall inverse participation ratio (inverse effective dimension) $D_\beta^{-1}$ is the weighted average: $D_\beta^{-1} = \sum_z \pdiagz D_\beta^{-1}(z)$.

\subsection{Porter-Thomas distribution over outcomes}
Theorem~\ref{thm:bitstring_PT} concerns the distribution of $\tilde{p}(z,t)$ for fixed $t$, treated as samples labeled by $z$. 
\begin{theorem}[Porter-Thomas distribution over outcomes] \label{thm:bitstring_PT}
Consider an initial state $\ket{\Psi_0}$ and a Hamiltonian $H$ satisfying the $2k$-th \textit{no-resonance condition}.
At a typical time $t\in [0,\infty)$, the weighted $k$-th moment of $\tilde{p}(z,t)$ approaches a universal value: $P^{(k)}(t) \equiv \sum_z \pdiagz \tilde{p}(z,t)^k = k! +  O(D_\beta^{-1/2})$, where $D_\beta^{-1} \equiv \sum_{z,E} \abs{\braket{z}{E}}^4 \abs{\braket{E}{\Psi_0}}^4/\pdiagz$ is the inverse effective Hilbert space dimension associated with the quench evolution, and $\{\ket{E}\}$ are the eigenstates of $H$.
\end{theorem}

\begin{proof}
Our proof strategy is the following. We first compute the \textit{time-averaged} value $\Ebb_t\L[P^{(k)}(t)\R]$. From Theorem~\ref{app:Hamil_k_design}, this quantity is $k! + O(D_\beta^{-1})$. We next compute the fluctuations of $P^{(k)}(t)$ over time, and we will find that $\text{std}_t\L[P^{(k)}(t)\R] = O(D_\beta^{-1/2})$. Combining these two results leads to our desired conclusion: $P^{(k)}(t) = k! + O(D_\beta^{-1/2})$ with high probability.

In Fig.~\ref{fig:thm_derivation}(a), we illustrate this computation of $\Ebb_t\L[ P^{(k)}(t)\R]$ for the case $k=2$. In general, we have:
\begin{equation}
    \abs{\Ebb_t\L[ P^{(k)}(t)\R]-k!} \leq \sum_z \pdiagz \abs{\Ebb_t\L[ \tilde{p}(z,t)^k\R]-k!} \leq {k \choose 2} k! \sum_z \pdiagz  D_\beta^{-1}(z) = O(D_\beta^{-1})
\end{equation}
Up to this point, we have only used the $k$-th Hamiltonian twirling identity. To estimate the temporal fluctuations of $P^{(k)}(t)$, we need the $2k$-th twirling identity.
\begin{align}
    \text{var}_t[P^{(k)}(t)] &= \mathbb{E}_t[P^{(k)}(t)^2] - \mathbb{E}_t[P^{(k)}(t)]^2 = \sum_{z,z'} \pdiagz \pdiag{z'} \frac{\Ebb_t[p(z,t)^k p(z',t)^k]-\Ebb_t[p(z,t)^k] \Ebb_t[p(z',t)^k]}{\pdiagz^{k}\pdiag{z'}^{k}} \label{eq:temporal_variance_calc}
\end{align}
By using the Hamiltonian $k$-th and $2k$-th twirling identities, any term in $\mathbb{E}_t[P^{(k)}(t)^2]$ that does not cancel with a term in $\mathbb{E}_t[P^{(k)}(t)]^2$ must involve a permutation between $z$ and $z'$ [illustrated for $k=2$ in Fig.~\ref{fig:thm_derivation}(b)]. As in our proof of Theorem~\ref{thm:temporal_PT}, without loss of generality, we let the energies corresponding to this permutation be $E_1$ and $E_2$, and multiply our overall expression by a factor of $2{2k \choose 2}$ (here, $E_1$ and $E_2$ are inequivalent). Then their sum is of the form
\begin{gather}
    \text{var}_t[P^{(k)}(t)] = \sum_{z,z'}\pdiagz \pdiag{z'} 2{2k \choose 2} \frac{\sum_{E_1,E_2} \abs{\braket{E_1}{\Psi_0}}^2\braket{z}{E_1}\!\braket{E_1}{z'} \abs{\braket{E_2}{\Psi_0}}^2\braket{z'}{E_2}\!\braket{E_2}{z}}{\pdiagz \pdiag{z'}} (\cdots) \\
    \equiv \sum_{z,z'}\pdiagz \pdiag{z'} 2{2k \choose 2} \frac{\bra{z}\rho_d\ketbra{z'}{z'}\rho_d\ket{z}}{\pdiagz \pdiag{z'}} (\cdots)~,\label{eq:temp_var_proof_line_above}\\
    (\cdots) = \sum_{\text{perms.}}\prod_{j=3}^{2k} \frac{\sum_{E_j} \abs{\braket{E_j}{\Psi_0}}^2\braket*{z_j^L}{E_j}\!\braket*{E_j}{z_j^R}}{\sqrt{\pdiag{z_j^L}\pdiag{z_j^R}}} \equiv \sum_{\text{perms.}} \prod_{j=3}^{2k} \frac{ \bra*{z_j^L}\rho_d\ket*{z_j^R}}{\sqrt{\pdiag{z_j^L}\pdiag{z_j^R}}}\leq \sum_{\text{perms.}} 1 \leq (2k-2)! \label{eq:temp_var_proof}
\end{gather}
where $z_j^L$ and $z_j^R$ can either be $z$ or $z'$, such that there exactly $k-1$ of the $\{z_j^L\}$ are $z$ and the other $k-1$ are $z'$, and likewise for the $\{z_j^R\}$. For brevity, we have also rewritten our expressions with $\rho_d \equiv \sum_E \abs{\braket{E}{\Psi_0}}^2 \ketbra{E}{E}$, the commonly studied \textit{diagonal ensemble}~\cite{linden2009quantum,short2011equilibration}. Note that $\pdiagz = \bra{z}\rho_d \ket{z}$. 

The inequalities in~\eqref{eq:temp_var_proof} are obtained as follows. The term involving $\rho_d$ is bounded by a simple application of the Cauchy-Schwarz identity, gives $\abs{\bra{z}\rho_d\ket{z'}}^2 \leq \pdiagz \pdiag{z'}$. Since 
$\abs{\bra{z}\rho_d\ket{z}}^2 \equiv \pdiagz \pdiag{z'}$, it follows that $ \abs{\bra*{z_j^L}\rho_d\ket*{z_j^R}}/\sqrt{\pdiag{z_j^L}\pdiag{z_j^R}} \leq 1$ regardless of the values of $z_j^{L}$ and $z_j^{R}$.  Finally, the sum $\sum_{\text{perms.}}$ is a sum over allowed permutations of a given list of energies $(E_1,\cdots ,E_{2k})$; we bound this by $(2k-2)!$. (Formally, $\sum_{\text{perms.}}$ should be \textit{inside} the summation over energies, with different permutations giving rise to different realizations of $\{z_j^L,z_j^R\}$. However, we use the notation above because this sum only gives a constant factor.) 

Combining~\eqref{eq:temp_var_proof_line_above} and ~\eqref{eq:temp_var_proof}, we obtain
\begin{align}
     \text{var}_t[P^{(k)}(t)] = (2k)! \sum_{z,z'}\pdiagz \pdiag{z'} \frac{\bra{z}\rho_d\ketbra{z'}{z'}\rho_d\ket{z}}{\pdiagz \pdiag{z'}} = (2k)!\tr(\rho_d^2) \leq (2k)!D_\beta^{-1} 
\end{align}
In our last inequality, we used $\tr(\rho_d^2) \leq D_\beta^{-1}$. This is derived from Sedrakyan's inequality, which states that for positive $u_i$ and $v_i$, $ \L(\sum_i u_i\R)^2/\sum_i v_i \leq \sum_i \L(u_i^2/v_i\R)$. In our context,
\begin{equation}
 \tr(\rho_d^2) \equiv \sum_E \abs{\braket{\Psi_0}{E}}^4 = \sum_E \abs{\braket{\Psi_0}{E}}^4 \frac{\L(\sum_z \abs{\braket{z}{E}}^2\R)^2}{\sum_z \pdiagz} \leq \sum_{z,E} \frac{\abs{\braket{z}{E}\braket{E}{\Psi_0}}^4}{\pdiagz} \equiv D_\beta^{-1}~.
 \label{eq:purity_is_less_than_eff_dim}
\end{equation}

This completes our proof that for a typical $t$, $P^{(k)}(t) = k! + O(D_\beta^{-1/2})$. 
\end{proof}

Our theorem states that when $D_\beta$ is sufficiently large, the $k$-th moment of $\{p(z,t)\}$ is close to $k!$. We can then conclude that $\{p(z,t)\}$ follows the exponential distribution when $D_\beta$ is sufficiently large. Our numerical results support this: fixing a late enough time $t$ (we numerically find that a time scaling linearly in system size is sufficient for our purposes) and plotting the histogram of the rescaled probabilities $\{\tilde{p}(z,t)\}_z$ (weighted by $\pdiagz$) shows a good agreement with the exponential distribution, as seen in Fig.~\ref{fig:theorem_graphic}(f).

Theorems~\ref{thm:temporal_PT} and \ref{thm:bitstring_PT} together suggest an analogue to the ergodic theorem in classical dynamical systems. Consider an ensemble of particles in a classical dynamical system, such as a billiard stadium. The ergodic theorem states that the spatial distribution of this ensemble at a fixed time is the same as the distribution of a fixed particle over time. In our context, the distributions of $\tilde{p}(z,t)$ are asymptotically the same in two settings: with fixed $t$ over configurations $z$ and with fixed $z$ over $t$. Furthermore, both distributions limit to the universal Porter-Thomas distribution.

\section{Detailed performance analysis for long quench dynamics} 
\label{app:Fid_FXEB_Fd}
In this section we apply the above results to characterize the accuracy of our fidelity estimator $F_d$. Our analysis is valid when the state has reached global equilibrium, quantified by the quantity $P^{(2)}(t) \equiv \sum_z \pdiagz \tilde{p}(z,t)^2$ reaching its equilibrium value of 2. In Theorem~\ref{thm:bitstring_PT}, we have bounded the deviation of the $k$-th moments $P^{(k)}(t)$ from its ideal value of $k!$. Under an approximation [Eq.~\eqref{eq:tildeq_approx}], our bound for $P^{(2)}(t)$ is sufficient to claim that the $F_d$ formula approximates $F$ with accuracy that is exponentially good in system size. By keeping the simplest energy-conservation terms that are the leading-order corrections in Theorem~\ref{thm:bitstring_PT}, we make quantitative predictions for the systematic error $\delta_\textrm{sys}$ and temporal fluctuations $\delta_\textrm{temp}$.
We explicitly confirm our theoretical analysis via numerical simulations presented in Fig.~3 of our main text and Fig.~\ref{fig:performance_metrics}, which shows excellent agreement between theory prediction and numerical results. 
This suggests that our formalism presented here is able to capture the performance of the $F_d$ formula to first order in $D_\beta^{-1}$, at least for the system sizes investigated.

We define the following relevant quantities:
\begin{align}
    \Fidd(t) &\equiv \sum_z \pdiagz \tilde{p}(z,t)^2 - 1 = P^{(2)}(t) - 1~,\\
    \FXEBd(t) &\equiv \sum_z q(z,t) \tilde{p}(z,t) -1~\label{eq:FXEB_def}.
\end{align}
$F_d$ is given by 
\begin{equation}
 F_d =  2\frac{\sum_z q(z) \tilde{p}(z)}{\sum_z p(z) \tilde{p}(z)} - 1 = 2\frac{\FXEBd+1}{\Fidd+1} -1   
\end{equation}
$\Fidd$ is a measure of whether the state has reached global equilibrium: after equilibration, $\Fidd \approx 1$, while at very short times, $\Fidd \sim O(D)$ is exponentially large when the initial state is a product state or a lowly-entangled state. $\FXEBd$ is analogous to the linear cross-entropy benchmark $F_\text{XEB}$.

We can bound $\Fidd \geq 0$ and $\FXEBd \geq -1$, and we obtain (weak) bounds on $F_d$:
\begin{align}
    -1 \leq &F_d \leq 2\sqrt{\frac{\sum_z \pdiagz \tilde{q}(z)^2}{ \sum_z \pdiagz \tilde{p}(z)^2}} - 1~,
\end{align}
We compute the average values over time of $F_d$ as well as its variance over time. We find that the typical variance is exponentially small in system size, therefore sampling at a single late time will give values close to its mean.

We use ``exponentially small in system size" as a shorthand to mean that a quantity is proportional to a measure of the effective dimension: $D_\beta^{-1}$ or $\tr(\rho_d^2)$. These are related by the following inequalities:
\begin{align}
D^{-1} \leq \sum_z \pdiagz^2 \leq \tr(\rho_d^2) \leq D_\beta^{-1} \leq 1~.
\end{align}
The first inequality is a straightfoward inequality between the arithmetic and quadratic means. The second inequality is obtained immediately by noting that 
$\tr(\rho_d)^2 = \sum_{z} \pdiagz^2 + \sum_{z\neq z'} \abs{\bra{z}\rho_d\ket{z'}}^2 \geq \sum_{z} \pdiagz^2$
, and the final inequality was established in Eq.~\eqref{eq:purity_is_less_than_eff_dim}. $D_\beta^{-1}$ is in general only upper bounded by $1$, since choosing $\ket{z} = \ket{E} = \ket{\Psi_0}$, gives $D_\beta^{-1} = 1$. However, for generic initial states, we expect $D_\beta^{-1}$ to be exponentially small in system size.

\subsection{Ansatz for how speckle patterns change under errors}
In the long-time limit, it is fruitful to approximate $q(z)$ as:
\begin{equation}
    q(z) \approx F p(z) + (1-F) p_\perp(z)~, \label{eq:tildeq_approx}
\end{equation}
where $p_\perp(z)$ is a distribution uncorrelated with $p(z)$. In general, in the case of incoherent errors, $\tilde{p}_\perp(z) \equiv p_\perp(z)/\pdiagz$ does not follow a PT distribution. 

We assume that $p(z)$ and $p_\perp(z)$ are uncorrelated. Specifically, we assume that:
\begin{align}
    \sum_z p_\perp(z) \tilde{p}(z)^k &\approx \mathbb{E}_z[p_\perp(z)] \sum_{z'}\tilde{p}(z')^k  \approx k!~, \label{eq:p_pperp_approx}
\end{align}
This illustrates that we do not require $\tilde{p}_\perp(z)$ to follow the Porter-Thomas distribution --- the energy density of the perturbed state can be quite different from the target state, as is the case when many errors occur. As an example, applying this on Eq.~\eqref{eq:FXEB_def} gives the desired $\FXEBd \approx F$. 

A related approximation applicable to higher order correction terms, is to assume a coherent error $V$, and write the experimental state as $\ket{\psi'(t)} = V(t) \ket{\psi_0(t)}$. $V(t)$ is unitary, and therefore $V^\dagger(t) V(t) = \mathbb{I}$. However, we replace instances of $V(t)$ or $V^\dagger(t)$ that do not cancel with their Hermitian conjugates with
\begin{equation}
 V(t), V^\dagger(t) \rightarrow \sqrt{F}~.    \label{eq:a_crude_approximation} 
\end{equation}
While crude, this approximation is useful to make quick estimates. Furthermore, its predictions quantitatively agree with our numerical results, indicating that the dominant sources of error are due to the structure of the $F_d$ formula, which we discuss below. 

\subsection{Systematic error} \label{sec:systematic_error}
Here, we characterize the systematic error $\delta_\text{sys.} \equiv \mathbb{E}_t \L[F_d(t)\R] - F$: the difference between the time-averaged $F_d$ and the true fidelity $F$ after an isolated error. To do so, we first study the time average of the simpler quantities $\Fidd(t) \equiv \sum_z \pdiagz \tilde{p}(z,t)^2-1$ and $\FXEBd(t) \equiv \sum_z \pdiagz \tilde{p}(z,t)\tilde{q}(z,t) -1$.

The $k$-th moment of the PT distribution (Theorems~\ref{thm:temporal_PT},\ref{thm:bitstring_PT}) immediately allows us to estimate $\Fidd$:
\begin{align}
\mathbb{E}_t \L[\Fidd(t)\R] &= \mathbb{E}_t \L[\sum_z \pdiagz \tilde{p}^2(z,t) - 1\R] = 1-  D_\beta^{-1}~.
\label{eq:aveFidd}
\end{align}
Using \eqref{eq:a_crude_approximation} gives
\begin{align}
    \FXEBd \approx F (1- D_\beta^{-1})~. \label{eq:FXEBd_with_subleading}
\end{align}

We would like to use these results to obtain the time-averaged value of $F_d$. The time average of the numerators and denominators $\FXEBd$ and $\Fidd$ should be done together. We can estimate its effects with the following Taylor expansion:
\begin{align}
    \Ebb\L[\frac{X}{Y}\R] \approx \frac{\Ebb[X]}{\Ebb[Y]} - \frac{\text{cov}[X,Y]}{\Ebb[Y]^2} + \frac{\text{var}[Y] \Ebb[X]}{\Ebb[Y]^3}~,\label{eq:taylor_expansion}
\end{align}
which gives the estimate for the systematic error in the presence of a single error: 
\begin{align}
    \delta_\textrm{sys} = \Ebb_t\L[F_d\R] - F &\approx (1-F) \L[\frac{D_\beta^{-1}}{2} + \frac{\text{var}_t\L[\Fidd\R]}{8}\R]~.~ \label{eq:var_t_prediction}
\end{align}
This prediction agrees well with numerical simulations in Fig.~\ref{fig:performance_metrics}(a) and Fig.~3(b) in the main text.

\subsubsection{Argument based on Eigenstate Thermalization Hypothesis}
In the limit of low error rates and long-time evolution, we can show that $F_d \approx F$, independent of our approximation for $\tilde{q}$ \eqref{eq:tildeq_approx}. To do so, we consider the effect of a single error $V$ that occurs at a variable time $t$. In the limit of long evolution, $F_d \approx \FXEBd  \equiv \sum_z \pdiagz \tilde{q}(z) \tilde{p}(z) -1$ since the denominator in $F_d$ approaches 2. Therefore, we will show that $\FXEBd \approx F$ for a single error averaged over the error occurrence time.
In the presence of multiple errors, the above relation will hold, as long as the errors are sufficiently far apart in time. We average over two quantities: (i) the time $t$ at which the error $V$ occurs, and (ii) the time $\tau$ after the error happens when the measurement is made. At this time, the ideal state is $\ket{\Psi(t+\tau)}$, while the perturbed state is $V(\tau)\ket{\Psi(t+\tau)}$, with $V(\tau) \equiv \exp(-iH \tau) V \exp(iH \tau)$ the Heisenberg evolution of $V$. Here, the fidelity $F(t) = \abs{\bra{\Psi(t+\tau)}V(\tau)\ket{\Psi(t+\tau)}}^2 = \abs{\bra{\Psi(t)}V\ket{\Psi(t)}}^2$ is independent of $\tau$ and only depends on $t$ (the time when the error is applied), while the XEB depends on both $t$ and $\tau$. We will average both $t$ and $\tau$. We first average over $t$ to compute the \textit{average} value of $\FXEBd$ at a time $\tau$ after an error is applied:
\begin{align}
    &\Ebb_t\L[\FXEBd(t,\tau)\R] = \sum_z \frac{\abs{ \bra{z}V(\tau)\rho_d \ket{z}}^2}{\pdiagz} - \sum_z \frac{\bra{z,z}\L(V^\dagger(\tau)\otimes\id\R)\rho_d^{(2)}\L(V(\tau)\otimes\id\R)\ket{z,z}}{\pdiagz}~,\label{eq:ave_FXEBd}
\end{align}
where we have defined the \textit{diagonal ensemble} $\rho_d$ and its analog $\rho_d^{(2)}$ over $\mathcal{H}^{\otimes 2}$ as follows:
\begin{align}
    \rho_d &\equiv \sum_E \abs{\braket{E}{\Psi_0}}^2 \ketbra{E}{E} = \Ebb_t \L[\ketbra{\Psi(t)}{\Psi(t)}\R]\\
    \rho_d^{(2)} &\equiv \sum_E \abs{\braket{E}{\Psi_0}}^4 \ketbra{E,E}{E,E}
\end{align}
$\rho_d^{(2)}$ is an analog of $\rho_d$ defined over two-copies of the Hilbert space. The first term in Eq.~\eqref{eq:ave_FXEBd} is the dominant contribution, while the second term comes from the energy-conservation term in Eq.~\eqref{eq:Ham_2_design}.

We next average over the time $\tau$ after the error, in order to claim that $\Ebb_t\L[ \FXEBd(t,\tau)\R] \approx \Ebb_t\L[F(t)\R]$. The averaging over $\tau$ gives:
\begin{align}
    \mathbb{E}_\tau \mathbb{E}_t\L[\FXEBd(t,\tau)\R] &= \sum_z \frac{\bra{z} \rho_d V_d \ketbra{z}{z} \rho_d V^\dagger_d \ket{z}}{\pdiagz} -\sum_z \frac{\bra{z,z} \rho^{(2)}_d \L(V_d \otimes V_d^\dagger\R) \ket{z,z}}{\pdiagz}~,\label{eq:double_ave_FXEBd}
\end{align}
where $V_d \equiv \sum_E \ketbra{E}{E}V\ketbra{E}{E}$ is the error operator \textit{dephased} in the energy basis, or the time-averaged operator $\Ebb_\tau V(\tau)$. We interpret the first term --- the dominant contribution to \eqref{eq:double_ave_FXEBd} --- as the average fidelity of states $\ket{\tilde{z}} \propto \sqrt{\rho_d}\ket{z}$ after the error operator $V_d$:
\begin{equation}
    \sum_z \frac{\bra{z} \rho_d V_d \ketbra{z}{z} \rho_d V^\dagger_d \ket{z}}{\pdiagz} = \sum_z \frac{\bra{z} \sqrt{\rho_d} V_d \sqrt{\rho_d} \ketbra{z}{z} \sqrt{\rho_d} V^\dagger_d \sqrt{\rho_d} \ket{z}}{\pdiagz} = \sum_z \pdiagz \abs{\bra{\tilde{z}} V_d \ket{\tilde{z}}}^2~,
\end{equation}
$\ket{\tilde{z}} \equiv \sqrt{\rho_d}\ket{z}/\sqrt{\pdiagz}$ is a normalized wavefunction, after a transformation of a  computational basis state. We interpret this transformation as an imaginary time evolution generated by the modular Hamiltonian of $\rho$. We now interpret $F_d$ as an average fidelity $\abs{\bra{\tilde{z}}V_d\ket{\tilde{z}}}^2$ of an error operator $V_d$ is applied to $\ket{\tilde{z}}$.

We want to relate $\mathbb{E}_\tau \mathbb{E}_t\L[\FXEBd(t,\tau)\R]$ to the fidelity $\Ebb_t \L[F(t)\R]$ averaged over the time $t$ that the error occurred. This is given by:
\begin{align}
    \Ebb_t \L[F(t)\R] &= \Ebb_t \L[\abs{\bra{\Psi(t)}V\ket{\Psi(t)}}^2\R] = \abs{\tr(\rho_d V)}^2 + \tr(\rho_d V\rho_d V^\dagger) - \tr[\rho_d^{(2)} (V\otimes V^\dagger)]~, \label{eq:ave_F}
\end{align}
with the terms ordered by (expected) magnitude.

In order to relate Eq.~\eqref{eq:ave_FXEBd} and \eqref{eq:ave_F}, we make two observations: (i) for a product initial state $\ket{\Psi_0}$ and local Hamiltonian $H$, $\rho_d$ approaches a microcanonical state since its energy variance vanishes relative to the size of the spectrum. This is because $H$ is local, implying that $\tr(H^2) \propto N^2$. Therefore, the eigenvalues of $H$ are spread over an energies scaling with total system size $N$. Meanwhile, (ii) by the locality of $H$ the uncertainty in energy of a product state is of order $\sqrt{N}$: $\Delta E^2 \equiv \bra{\Psi_0} H^2 \ket{\Psi_0} - \bra{\Psi_0} H \ket{\Psi_0}^2 \propto N$. We model the distribution of energies as a Gaussian distribution with mean $\bar{E} \equiv \bra{\Psi_0} H \ket{\Psi_0}$ and variance $\Delta E^2$: 
\begin{equation}
\sum_{E=\bar{E}-dE/2}^{\bar{E}+dE/2}\abs{\braket{E}{\Psi_0}}^2 \approx \rho\L(\bar{E}\R)\overline{\abs{\braket{E}{\Psi_0}}^2} dE \approx \frac{1}{\Delta E \sqrt{2 \pi}}\exp(-\frac{(E-\bar{E})^2}{2(\Delta E)^2}) dE~, \label{eq:Gaussian_ETH_estimate}
\end{equation}
where $\rho(\bar{E})$ is the density of states around $\bar{E}$. We define this such that $\int \rho(E) dE = D$, such that $\rho(\bar{E}) = \exp(cN)$ for some $c>0$. Eq.~\eqref{eq:Gaussian_ETH_estimate} also gives the estimate $\sum_{E=\bar{E}-dE/2}^{\bar{E}+dE/2} \abs{\braket{E}{\Psi_0}}^4 \propto 1/[\rho(\bar{E}) \Delta E] \sim \exp(-cN)$ (omitting polynomial factors of $N$).

 (iii) Lastly, the Eigenstate Thermalization Hypothesis (ETH) states that the diagonal elements $\bra{E}V\ket{E}$ of a local operator $V$ are well described by a smooth function of $E$, $\bra{E}V\ket{E} \approx v(E)$. The off-diagonal elements are exponentially smaller: $\abs{\bra{E}V\ket{E'}}^2 \approx f(E-E')\langle V^2 \rangle_\text{micro} \exp(-S(\bar{E})/2)$, for some function $f$ which is $O(1)$ when the separation $E-E'$ is small, and $\bar{E} = (E+E')/2$. $\langle V^2 \rangle_\text{micro}$ is the microcanonical value of $V^2$ near $\bar{E}$, and $S(\bar{E}) = -\ln(\rho(\bar{E}))$ is the thermodynamic entropy~\cite{deutsch2018eigenstate}.
Then the expression \eqref{eq:ave_F} for the fidelity can be approximated as follows: 
\begin{align}
    \Ebb_t \L[F\R] &~= \abs{\tr(\rho_d V)}^2 + \tr(\rho_d V\rho_d V^\dagger) - \tr(\rho_d^{(2)} (V\otimes V^\dagger))\\
    &~= \big\vert\sum_E \bra{E}V\ket{E} \abs{\braket{E}{\Psi_0}}^2 \big\vert^2 + \sum_{E,E'}\bra{E}V\ketbra{E'}{E'}V^\dagger\ket{E} \abs{\braket{E}{\Psi_0}}^2 \abs{\braket{E'}{\Psi_0}}^2  - \sum_{E}\abs{\bra{E}V\ket{E}}^2 \abs{\braket{E}{\Psi_0}}^4 \nonumber\\
    &\overset{\text{ETH}}{\approx} \abs{v(\bar{E})}^2 +  \langle V^2 \rangle_\text{micro} \exp(-c'N)~,
\end{align}
for some constant $c'>0$.  

Meanwhile, to analyze $\mathbb{E}_\tau \mathbb{E}_t\L[\FXEBd(t,\tau)\R]$, we approximate $\abs{\bra{z}\rho_d V_d \ket{z}} = \sum_E \abs{\braket{z}{E}\braket{E}{\Psi_0}}^2 \abs{\bra{E}V\ket{E}}  \overset{\text{ETH}}{\approx} \overline{\bra{E}V\ket{E}}_{\bar{E}} \sum_E \abs{\braket{z}{E}\braket{E}{\Psi_0}}^2 = \vert v(\bar{E})\vert  \pdiagz$. Then
\begin{align}
    \mathbb{E}_\tau \mathbb{E}_t\L[\FXEBd(t,\tau)\R] &~= \sum_z \frac{\bra{z} \rho_d V_d \ketbra{z}{z} \rho_d V^\dagger_d \ket{z}}{\pdiagz} -\sum_z \frac{\bra{z,z} \rho^{(2)}_d \L(V_d \otimes V_d^\dagger\R) \ket{z,z}}{\pdiagz}~,\label{eq:double_ave_FXEBd}\\
    &~= \sum_z \frac{1}{\pdiagz} \Big[\vert\sum_{E} \abs{\braket{z}{E} \braket{E}{\Psi_0}}^2 \bra{E}V\ket{E}\vert^2 - \sum_E  \abs{\braket{z}{E}\braket{E}{\Psi_0}}^4 \abs{\bra{E}V\ket{E}}^2 \Big]\\
    &\overset{\text{ETH}}{\approx} \sum_z \frac{\abs{v(\bar{E})}^2}{\pdiagz} \Big[\Big( \sum_{E} \abs{\braket{z}{E} \braket{E}{\Psi_0}}^2 \Big)^2 - \sum_E  \abs{\braket{z}{E}\braket{E}{\Psi_0}}^4 \Big]\\
    &~= \abs{v(\bar{E})}^2(1-D_\beta^{-1})
\end{align}

Since $D_\beta^{-1}$ is also typically exponentially small in system size, we conclude that $\abs{\mathbb{E}_\tau \mathbb{E}_t\L[\FXEBd(t,\tau)\R] - \Ebb_t \L[F(t)\R] } = O(\exp[-dN])$, for some $d>0$. This is our desired result, obtained here through the ETH. 

{\subsubsection{Extensions beyond local errors}
The above argument shows that in the presence of local coherent errors, $F_d$ faithfully estimates $F$ in the long time limit. This can be immediately extended to incoherent errors by averaging over coherent errors. By unraveling an error channel into quantum trajectories, these results also apply to stochastic, local error channels, as long as the error rate is sufficiently low (numerically investigated in Fig.~\ref{fig:performance_metrics}). 

Finally, we argue that these results also apply to coherent Hamiltonian errors, i.e.~quench evolution by different Hamiltonian parameter, utilized in our process benchmarking examples. Time evolution under a slightly perturbed local Hamiltonian can be expanded into a coherent sum of paths representing evolution by the original Hamiltonian, weakly interspersed by the perturbation term. This is similar to a stochastic error channel, except we can have coherent interference between different paths. We argue that this interference effect from off-diagonal terms is negligible: different paths have vanishing overlap and essentially uncorrelated speckle patterns, and their sum is further suppressed by their random phases. An analogous argument applies to target state benchmarking within a ground state phase diagram [Fig.~\ref{fig:target_state_benchmarking}], which can be understood as applying imaginary-time evolution (cooling) between states. This effect is perturbative and can also be decomposed into paths of local excitations, as long as the two states are sufficiently nearby.}

\subsection{Variance over time}
The variance over time utilizes the fourth Hamiltonian twirling identity. The proof of Theorem~\ref{thm:bitstring_PT} estimates this as $O(\tr[\rho_d^2])$. We explicitly compute the prefactors of the expected leading order contribution arising from the simplest energy-conserving terms, illustrated in Fig.~\ref{fig:thm_derivation}(b):
\begin{align}
    \text{var}_t\L[\Fidd\R] &= \Ebb_t \L[\Fidd(t)^2\R] - \Ebb_t\L[\Fidd(t)\R]^2=  16\L[\cancelto{0}{\tr(\rho_d^2) - \tr(\rho_d^{(2)})}\R] + 4 \sum_{z,z'} \frac{\abs{\bra{z}\rho_d\ket{z'}}^4}{\pdiagz \pdiag{z'}}+\cdots~. 
\end{align}

\begin{figure}[tb]
    \centering
    \includegraphics[width=0.9\textwidth]{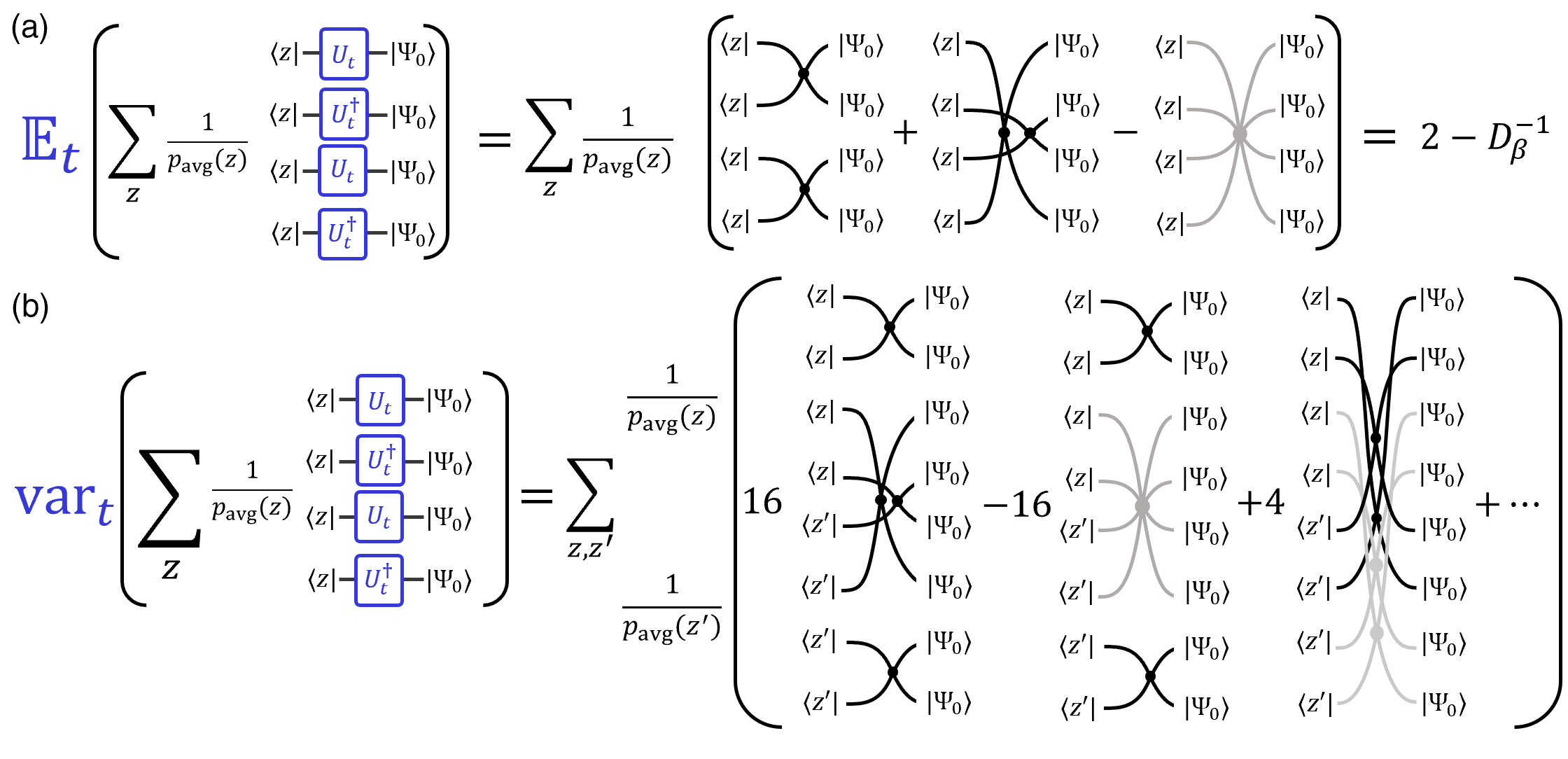}
    \caption{Diagrammatic derivation of terms that contribute to (a) $\Ebb_t\L[\Fidd(t)\R]$ and (b) $\text{var}_t\L[\Fidd(t)\R]$. In (a), the two permutations and collision term give $2-D_\beta^{-1}$. In (b), we list the (expected) leading order contributions. 16 out of the 24 possible permutations each give $\tr(\rho_d^2)$, these are cancelled out by corresponding collision terms: representative permutations and collision terms are illustrated here. Another four permutations give the leading order contribution to $\text{var}_t\L[\Fidd(t)\R] = 4 \sum_{z,z'} \abs{\bra{z}\rho_d\ket{z'}}^4/(\pdiagz \pdiag{z'}) \leq 4 \tr(\rho_d^2)$.}
    \label{fig:thm_derivation}
\end{figure}

\begin{figure}[tbp]
    \centering
    \includegraphics[width=0.9\textwidth]{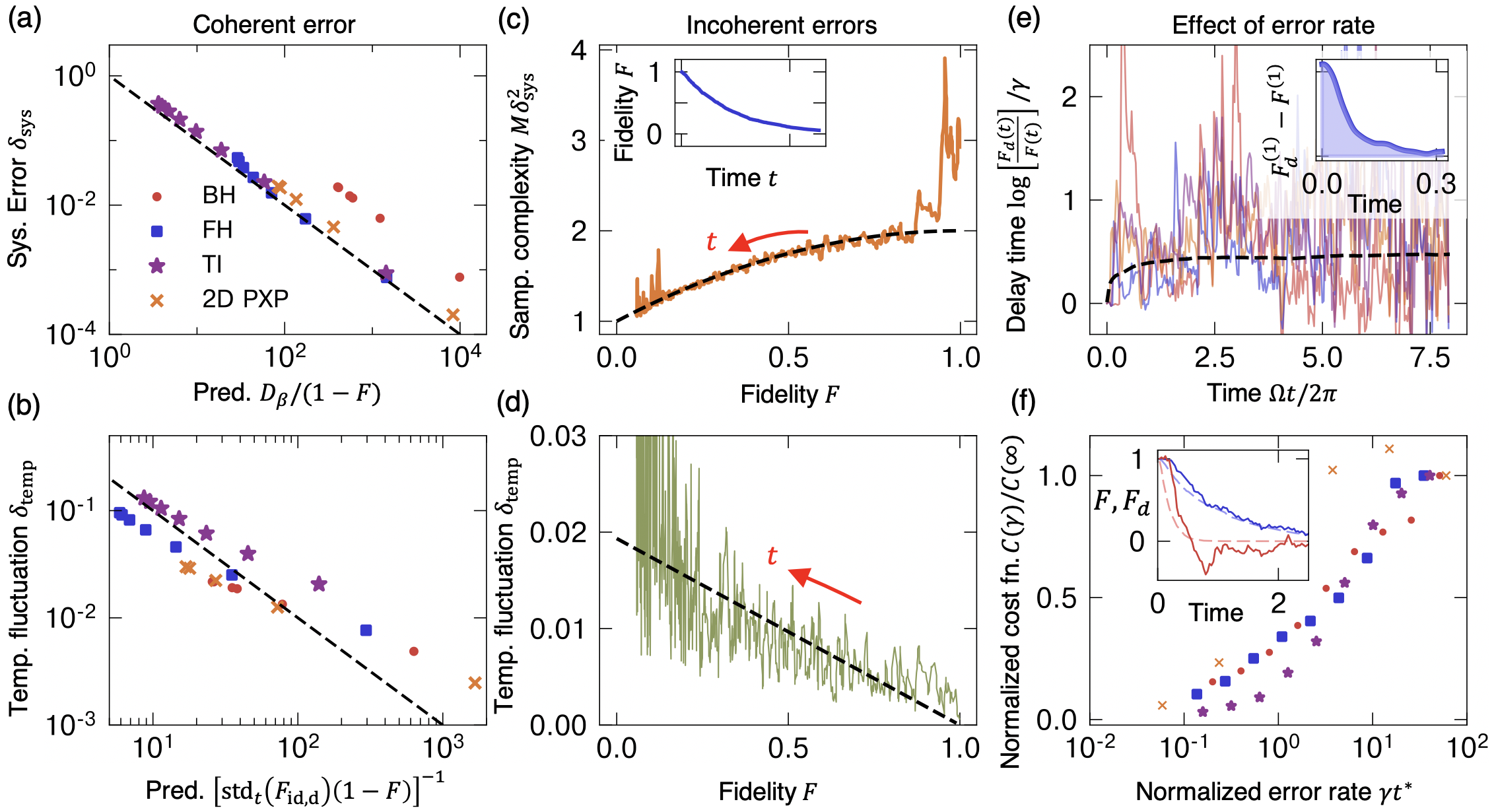}
    \caption{Performance metrics for $F_d$ under single and continuous errors. In all cases, our numerical results quantitatively agree with our theoretical predictions (black dashed lines), with all parameters numerically extracted from independent numerical data. (a,b) After a single error, the (a) systematic error $\delta_\text{sys.}$ and (b) temporal standard deviation $\delta_\text{temp.}$ quantitatively agrees with the ansatz $\delta_\text{sys.} = (1-F)D_\beta^{-1}$ and $\delta_\text{temp.} = \sqrt{5/4}(1-F) \text{std}_t\L[F_\text{id,d}\R]$, across all four systems studied: a Bose-Hubbard model (red), Fermi-Hubbard model (blue), trapped-ion spin chain  Hamiltonian (purple), and 2D Rydberg model (orange). (c,d) In the presence of continuous errors, $F$ decays with time from 1 to 0 (inset). Therefore, we can test our predictions for the dependence of our quantities of interest---the sample complexity (c), and the temporal fluctuations (d)---on the fidelity $F$ in a single time series of open system dynamics. We present such a time series for a Bose-Hubbard model with 8 bosons on 8 sites, where we plot the (c) sample complexity at each point in time and the (d) temporal fluctuations in a narrow time window against the fidelity $F$ at that point in time. After a short thermalization time, the sample complexity follows our universal prediction in Eq.~\eqref{eq:sample_complexity_prediction}: $M \delta_\text{stat.}^2 = (1+2F-F^2)$, and the temporal fluctuations satisfy Eq.~\eqref{eq:temp_var_est}: $\delta_\text{temp.} = (1-F) \text{std}_t\L[F_\text{id,d}\R]/2$. (e) Under continuous errors, $F_d(t)$ will appear to lag behind $F(t)$ with a delay time. In particular, the rescaled curves of $F_d/F$ for a range of error rates $\gamma=0.0005$ to $\gamma=0.005$ collapse, and are close to our prediction in Eq.~\eqref{eq:delay_time}, obtained from the single-error response curve (inset). (f) At low error rates, $F_d$ estimates $F$ well, but they deviate significantly at high error rates. Inset: Across all platforms, a normalized cost function $C(\gamma)/C(\infty)$~\eqref{eq:cost_function} shows a crossover when the average time between errors is less than an error autocorrelation time $t^*$~\eqref{eq:autocorrelation_time}.} 
    \label{fig:performance_metrics}
\end{figure}
 While this expression allows us to theoretically understand the expected the magnitude of $\text{var}_t\L[\Fidd\R]$, in practice we numerically evaluate $\text{var}_t\L[\Fidd\R]$ and investigate its relation to $\text{var}_t\L[F_d\R]$. In the presence of errors, the variance of $\FXEBd$ can be calculated using Eq.~\eqref{eq:p_pperp_approx}. This gives $\text{var}_t\L[\FXEBd \R] \approx F^2 \text{var}_t \L[\Fidd\R]$ for incoherent errors. In the case of a \textit{single, coherent error}, we must include the coherence of $\tilde{p}_\perp^2$: $\langle \tilde{p}_\perp^2 \rangle \approx 2$ and we obtain  $\text{var}_t\L[\FXEBd \R] \approx \L[F^2 + (1-F)^2\R] \text{var}_t \L[\Fidd\R]$ instead. In both cases, we also have $\text{cov}_t\L[\FXEBd,\Fidd\R] \approx F \text{var}_t \L[\Fidd\R]$. 

Performing the same Taylor expansion as in Eq.~\eqref{eq:taylor_expansion},
\begin{align}
    \text{var}\L[\frac{X}{Y}\R] \approx \frac{\Ebb[X]^2}{\Ebb[Y]^2} \L[\frac{\text{var}[X]}{\Ebb[X]^2} + \frac{\text{var}[Y]}{\Ebb[Y]^2} - 2 \frac{\text{cov}[X,Y]}{\Ebb[X] \Ebb[Y]}\R]\,
\end{align}
gives the estimate
\begin{align}
    \delta^2_\textrm{temp} = \text{var}_t\L[F_d\R] &\approx \frac{1}{4} (1-F)^2 \text{var}_t\L[\Fidd\R]~\text{(open system dynamics)}
    \label{eq:temp_var_est}\\
    &\approx \frac{5}{4} (1-F)^2 \text{var}_t\L[\Fidd\R]~\text{(coherent errors)}
\end{align}
Here, the larger fluctuations for a single coherent error are associated with the fact that it is pure. We numerically evaluate $\text{var}_t \L[\Fidd\R]^{-1}$, and compare our predictions for $\text{var}_t \L[F_d\R]$ in terms of $\text{var}_t \L[\Fidd\R]$ and $F$ in Fig.~3(c) of the main text and Fig.~\ref{fig:performance_metrics}(b,d) (respectively single error and open system dynamics).

\subsection{Sample complexity} \label{app:sample_complexity}
The sample complexity of $F_d$ is related to the variance of the estimator $\hat{F}_d$ with a finite number of samples $M$. The variance $\delta_\text{stat}^2$ is given by 
\begin{equation}
    \delta_\text{stat}^2\equiv\text{var}\big[\hat{F}_d\big] = \frac{4 \L(\sum_z q(z) \tilde{p}(z)^2  - \L[\sum_z q(z) \tilde{p}(z)\R]^2 \R)}{M\L[\sum_z \pdiagz \tilde{p}(z)^2\R]^2} \overset{\text{large } t}{\approx} \frac{1}{M}\L(\sum_z q(z) \tilde{p}(z)^2  - \L[\sum_z q(z) \tilde{p}(z)\R]^2 \R) 
\end{equation}
We estimate this by computing computing its time-averaged value, with the third Hamiltonian twirling identity, which gives $\Ebb_t \L[\sum_z \pdiagz \tilde{p}(z,t)^3\R] = 3!-(2!)^2 + \cdots \approx 2$. Our ansatz for $\tilde{q}(z,t)$~\eqref{eq:p_pperp_approx} gives:
\begin{align}
    \Ebb_t \L[M \delta^2_\text{stat} \R] = \Ebb_t\L[ \sum_z q(z,t)\tilde{p}(z,t)^2 - \L(\sum_z q(z,t)\tilde{p}(z,t)\R)^2\R]\approx 1+2F-F^2~. \label{eq:sample_complexity_prediction}
\end{align}
This non-trivial prediction is verified in Fig.~\ref{fig:performance_metrics}(c) and Fig.~3(d) of the main text.

\subsection{Moderate error rate: Effects of delay time $\tau_d$}
\label{app:delay time}
Here we argue that under open system dynamics with a moderate error rate, the ratio $F_d/F$ approaches a constant $F_d/F \approx \exp( \gamma \tau_d )$, where $\gamma$ is the \textit{total} error-rate and $\tau_d$ is a response time independent of system size. When the fidelity $F$ shows a simple exponential decay, the delay time leads to $F_d$ lagging behind $F$ by a constant time $\tau_d$. This is in contrast to the constant difference $F_d-F$ seen under a single error, which is exponentially small in system size. We quantify this delay time effect by introducing the notion of a response function of $F_d$ after a single error, $F_d^{(1)}$ , and making the assumption that the effects of multiple errors on $F_d$ are independent and factorize.

We assume that an error operator $X$ can occur at any time during the open system dynamics. Without loss of generality, we normalize $X$ so it does not change the norm of a typical state: $\bra{\psi(t)} X^\dagger X \ket{\psi(t)} \approx 1$. The fidelity $F$ as a function of time is:
\begin{align}
    F(t) &\approx p_{ne}(t) + \int_0^t d t_1 p_e(t_1) \abs{\langle X(t_1) \rangle}^2 + \int_0^t  d t_1 \int_{t_1}^t d t_2 p_e(t_1,t_2) \abs{\langle X(t_2) X(t_1) \rangle}^2 + \cdots \nonumber\\
    &\approx p_{ne} + \sum_{n=1}^\infty \frac{1}{n!} \L[\int_0^t d t_i p_e(t_i) \abs{\langle X(t_i) \rangle}^2 \R]^n = \exp[\L(F_\beta-1\R)\gamma t] \label{eq:Ft}~,
\end{align}
where we denote $p_{ne}(t) \propto \exp(-\gamma t)$ as the probability of not having an error from time $0$ to $t$ and $p_e(t_1,\cdots,t_n) \propto \gamma^n  \exp(-\gamma t)$ as the probability density of errors occurring at times $t_1,\cdots,t_n$. If $n$ errors occur during a trajectory, the pure state asspcoated with this trajectory will have a fidelity $\abs{\bra{\psi(t_n)} X\cdots X e^{-iH(t_2-t_1)} X \ket{\psi(t_1)}}^2$. As long as the error rate is low enough that errors are typically far apart in time, we assume that this expression factorizes: $F =  \abs{\langle X(t_n) \cdots X(t_1) \rangle}^2 \approx \abs{\langle X \rangle}^{2n}_\beta \equiv F_\beta^n$. Here the quantity $F_\beta \equiv \abs{\langle X \rangle}^2_\beta$ should be interpreted as the fidelity after a \textit{single error} $X$ on a typical thermal state. 

Meanwhile, we can estimate how $F_d$ behaves under open system dynamics in terms of its behavior in response to a single error. We denote $F^{(1)}_d(\tau)$ as the average $F_d$ value at time $\tau$ after a single error. 
\begin{align}
    F_d(t) &\approx p_{ne}(t) + \int_0^t d t_1 p_e(t_1) F_d(t-t_1) + \int_0^t \int_{t_1}^t d t_1 d t_2 p_e(t_1,t_2) F_d(t-t_2,t-t_1) + \cdots \nonumber\\
    &\approx p_{ne} + \sum_{n=1}^\infty \frac{1}{n!} \L[\int_0^t d t_i p_e(t_i) F_d^{(1)}(t-t_i)  \R]^n = \exp[\gamma \int_0^t d\tau \L(F_d^{(1)}(\tau)-1\R)] \label{eq:Fdt}
\end{align}
where we again assume that the value of $F_d$ after $n$ errors factorizes as: $F_d(t_1,\cdots,t_n)\approx \prod_i F_d^{(1)}(t-t_i)$.

Combining Eqs.~\eqref{eq:Ft},~\eqref{eq:Fdt}, we have a simple expression for the ratios:
\begin{equation}
    \frac{F_d(t)}{F(t)} \approx \exp[\gamma \int_0^t d\tau \left(F_d^{(1)}(\tau)-F_\beta\right)] \approx \exp[\gamma \tau_d]~, \label{eq:delay_time}
\end{equation}
where the quantity $\tau_d \equiv \int_0^\infty d\tau \left(F_d^{(1)}(\tau)-F_\beta\right) dt $ has a natural interpretation as the \textit{total area under the response curve} $F_d^{(1)}(\tau)-F_\beta$ (Fig.~\ref{fig:performance_metrics}(e) inset) --- if we ignore the (exponentially small) systematic error $\bar{F}_d^{(1)} - F$, this total area is a constant. This prediction is supported in Fig.~\ref{fig:performance_metrics}(e), where we observe the ratio $F_d(t)/F$ saturating at our predicted value for $\exp[\gamma \tau_d]$, with $\tau_d$ extracted from the corresponding single error curve $F_d^{(1)}(t)$.

\subsection{Large error rate: Critical error rate and operator scrambling time}
As the error rate increases, there is a critical rate at which $F_d$ shows significant deviations from $F$ [Fig.~\ref{fig:performance_metrics}(f) inset]. 
Naively, we would expect that this error rate is solely determined by the delay time discussed above.
However, this is not necessarily the case. The delay time is a property of the relaxation dynamics of a local perturbation occurring after the state has reached thermal equilibrium. When the error rate is substantially larger than the local relaxation time, the fidelity becomes very small even before the state has a chance to reach thermal equilibrium. Therefore, the delay time might not be the appropriate timescale for the critical error rate. Instead, we seek to characterize a different timescale relevant in this non-equilibrium dynamics. We hypothesize that a suitable timescale is the error operator \textit{autocorrelation time} --- intuitively an operator scrambling time --- $t^*$ with respect to the nonequilibrium initial state. Our numerical results support our hypothesis: we find that the critical error rate is related to  $t^*$.

We define the autocorrelation function of the error operator $X_j$, for a given initial state $\ket{\Psi_0}$ and quench Hamiltonian $H$, as:
\begin{equation}
    A_j(t) = \frac{\abs{\bra{\Psi_0} e^{iHt} X_j e^{-i H t} X_j \ket{\Psi_0}}^2}{\bra{\Psi_0} e^{iHt} X_j^2 e^{-i H t}\ketbra{\Psi_0}{\Psi_0}X_j^2\ket{\Psi_0}}~,
\end{equation}
where we have chosen this normalization such that $0 \leq A(t) \leq 1$. The states $X_j\ket{\Psi(t)}$ are generally not normalized because the error operators $X_j$ may not be unitary (for example, in the Bose-Hubbard model, the errors are associated with photon scattering, and we take $X_j$ to be $n_j$). We define the autocorrelation time as:
\begin{equation}
    t^*_j \equiv \frac{\int_0^\infty dt \abs{A_j(t)-A_j(\infty)}}{1-A_j(\infty)}~,
    \label{eq:autocorrelation_time}
\end{equation}
where $A_j(\infty)$ is the infinite time value $A_j(\infty) = \lim_{t\rightarrow \infty} A_j(t)$.
This can be analytically evaluated from the second Hamiltonian twirling identity [Eq.~\eqref{eq:Ham_2_design}]. We take the overall operator scrambling time as the average of $t^*_j$ over error operators $X_j$. This timescale is independent of system size, unlike other timescales such as the time for the entanglement entropy to saturate, which scales with system size.

Importantly, we find that the critical error rate is associated with the non-equilibrium dynamics of the initial state, as opposed to the dynamics after the state has reached thermal equilibrium. To numerically demonstrate this, we initialize the noisy dynamics on a time-evolved state $e^{-iHt}\ket{\Psi_0}$ for sufficiently large $t$, after the state has thermalized. We find that in the regime of error rates studied, there is no critical change in the accuracy of $F_d$, merely a constant delay time between $F_d$ and $F$, as predicted in Section~\ref{app:delay time}, which translates to a smooth increase in the relative error $F_d/F$.
This is analogous to Ref.~\cite{dalzell2021random,gao2021limitations}, in which the error rate in random unitary circuits was found to have the largest effect at the start and the end of the evolution. 

We define a cost function \begin{equation}
    C(\gamma) \equiv \int_0^T dt \abs{F_d(t)-F(t)}~, \label{eq:cost_function}
\end{equation}
that captures the total deviation between $F_d$ and $F$ over time. In Fig.~\ref{fig:performance_metrics}(f), we find that the rescaled cost function $C(\gamma)/C(\infty)$ shows a crossover from 0 to 1 at $\gamma t^* \sim 1$, indicating that $t^*$ is the relevant timescale for our setting. We choose such a normalization such that the curves from different simulator platforms can be plotted on the same graph: they demonstrate a cross over in a common regime $\gamma t^* \sim 1$.

\subsection{Symmetry Resolution}
In the main text, we described how $F_d$ naturally incorporates symmetries of the quench Hamiltonian. Here, we discuss this in more detail and provide an explicit example. If the initial state belongs to a single symmetry sector $s$, the quenched state will remain in that sector. Assume for simplicity that the state is Haar random in $s$. The simpler formulae $F_\text{XEB}$ and $F_c$ would be appropriate if our measurements $z$ were also states in $s$. However, this is generically not the case --- if the outcome $z$ does not have a well-defined symmetry quantum number, it should be weighted by its overlap with $s$. The factor $\pdiagz$ automatically weights $z$ appropriately, without explicit knowledge of the symmetry. 

As an illustration, consider spatial inversion symmetry $\mathcal{I}$ on a lattice. The even parity sector has basis states $\ket{e_z} = \mathcal{N}(\ket{z} + \ket{\bar{z}})$, where $\bar{z}$ is the mirrored configuration of $z$. If $\bar{z} = z$, $\mathcal{N} = 1/2$, otherwise $\mathcal{N} = 1/\sqrt{2}$. 
Consider the optimal case where $\ket{\Psi_0}$ is a random even-parity state and the dynamics is a random unitary that conserves parity. The universal statistics would be apparent if we could measure in the $\{\ket{e_z}\}$ basis --- a challenging task. Instead, outcome configurations $z$ have unequal average probabilities: $\pdiagz = 1/D_s$ if $z = \bar{z}$ and $\pdiagz = 1/(2D_s)$ otherwise, where $D_s$ is the dimension of the symmetry sector. While $p(z)$ does not follow the Porter-Thomas distribution, $\tilde{p}(z)$ does. Therefore, $F_d$ correctly estimates $F$, while $F_c$ unequally weights some basis states. This argument generalizes to other symmetries such as translational invariance. Since $\pdiagz$ automatically accounts for symmetries, our conclusions hold even in the presence of extensively many symmetries, as in the case of integrable systems. Finally, in a special case where the readout basis states have well-defined symmetry quantum number, $F_d$ does not estimate $F$. This is discussed in Section~\ref{sec:symmetry_nonergodicity}.

\subsection{Alternative fidelity estimator $F_e$}
\label{sec:Fe_summary}
In the following section, we use the random MPS model to derive an alternative fidelity estimator 
\begin{equation}
 F_e \equiv \frac{\sum_z q(z)\tilde{p}(z)-1}{\sum_z p(z) \tilde{p}(z)-1} = \frac{\FXEBd}{\Fidd}~.   
\end{equation}
Here, we summarize some of its properties. While the temporal fluctuations and sample complexity of $F_e$ are similar to those of $F_d$, the systematic error (after a long time) is smaller.

(i) Systematic error. The analysis in Section~\ref{sec:systematic_error} indicates that the systematic error after a single error $\delta_\text{sys} = \Ebb_t F_d(t) - F = O(D_\beta^{-1})$. Repeating the same analysis for $F_e$ indicates that the systematic error vanishes for $F_e$, to first order. This indicates that $F_e$ is less sensitive to small values of $D_\beta$, and hence may perform better in quench dynamics with low effective temperature. This is supported in Fig.~\ref{fig:failure_cases1}(e,f). More fundamentally, $F_e$ is more robust against deviations of the denominator $\sum_z \pdiagz \tilde{p}(z)^2$ from 2. For example, in the presence of particle hole symmetry, this denominator can be 3. $F_d$ can be appropriately modified to take this into account, but $F_e$ estimates the fidelity without any modification [Fig.~\ref{fig:failure_cases1}(a,b)].

(ii) Non-local errors. Our analysis based on random matrix product states indicates that $F_e$ will estimate $F$ for highly non-local errors. This is numerically supported in Fig.~\ref{fig:failure_cases2}(f). 

(iii) Short-time response. A tradeoff of $F_e$ is that it typically has a longer response time than $F_d$ to learn the fidelity $F$. This is demonstrated in Fig.~\ref{fig:failure_cases2}(e) and Fig.~\ref{fig:target_state_benchmarking_local}(b). 

\section{Detailed performance analysis for short quench dynamics}\label{app:random_MPS}

\subsection{Overview: random matrix product states}

\begin{figure*}[tb]
    \centering
    \includegraphics[width=0.9\textwidth]{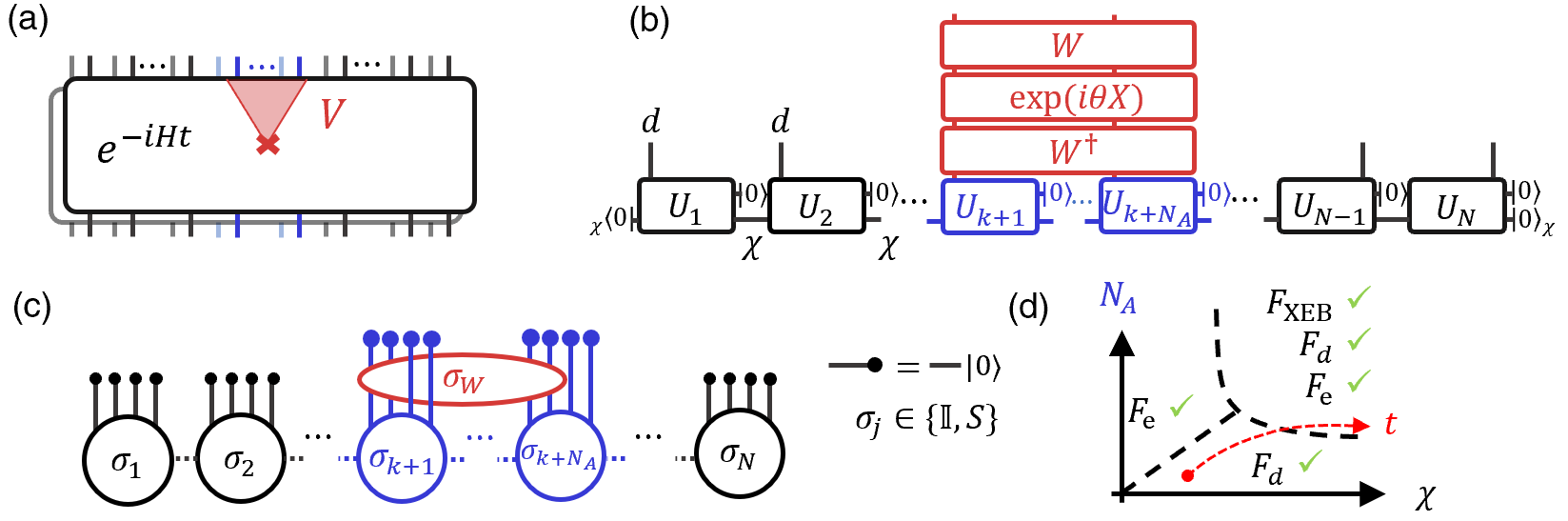}
    \caption{Random matrix product state model. (a) We consider two states evolved by a short quench. In one copy, an error has occurred, which spreads to a small subsystem $A$, and no error has occured in the other copy. (b) A toy model for the non-equilibrium state is the \textit{random MPS} --- a matrix product state with a random $\chi d \times \chi d$ matrix at every site, with $\chi$ the MPS bond dimension, and $d$ the local Hilbert space. The error $V$ is also modeled with a random unitary $W$ and an operator $\exp(i\theta X)$. (c) The expression for $F_d$, averaged over random unitaries is mapped to the partition function of an effective ferromagnetic Ising spin chain, with Ising degrees of freedom $\sigma_j, \sigma_W \in \{\id, \swap\}$. This partition function is given in Eq.~\eqref{eq:random_MPS_prediction}. (d) We predict regimes of validity for our benchmarking formulae $F_d, F_e$, and $F_\text{XEB}$, depending on the error operator size $N_A$ and MPS bond dimension $\chi$. A typical quench is always in the regime of validity of $F_c$ (red arrow).}
    \label{fig:random_MPS}
\end{figure*}
The infinite-time averaging in Section~\ref{app:Fid_FXEB_Fd} tells us that our fidelity estimators work at late times. Our formulae are also designed to work at early times, before the state has reached equilibrium. To capture this early time physics, we introduce a toy model, the \textit{random matrix product state}.

We model the state $\ket{\Psi_0(t)}$ as a matrix product state (MPS) with bond dimension $\chi$ and local Hilbert space dimension $d$. The MPS is constructed by choosing a $\chi d \times \chi d$ random unitary $U_j$ on each site $j$, as illustrated in Fig.~\ref{fig:random_MPS}(b). At every site, a $d$ degree of freedom is eliminated by contracting to the state $\ket{0}$. At the boundary sites, an additional $\chi$ degree of freedom is removed by contracting to a state $\ket{0}_\chi$ on the bond-dimension space. This provides a toy model for a state with entanglement $S \sim \log(\chi)$. 

We model the error operator $V$ as $V = W \exp(i \theta X) W^\dagger$, where $V$ is taken to be supported on some region $A$, and $W$ is a random unitary on $\mathcal{H}_A$. $\exp(i \theta X)$ is an operator whose precise form will be irrelevant. However, the parameter $\theta$ measures the strength of the coherent error, and can be used to tune the fidelity between the perturbed and initial states. The random unitary $W$ models the scrambling nature of $V(\tau)$ while making the calculation tractable using techniques of averaging over random unitaries~\cite{collins2016random}.
 
This toy model captures what we believe to be the essential physics: entanglement and operator growth. In particular, it is well known that both entanglement entropy and operator size grow linearly with time. This corresponds to:
\begin{align}
    \chi \sim \exp(v_S t)~,~N_A \sim v t~,
\end{align}
where $v_S$ is an entanglement velocity, and $v$ is an operator spreading velocity~\cite{nahum2018operator,khemani2018operator}.

Our calculations give us expressions for how $F_\text{XEB} \equiv D^{-1}\sum_z q(z) p(z)-1$ and $F_\text{id} \equiv D^{-1} \sum_z p(z)^2 - 1$, and the fidelity $F$ depend on $N_A$ and $\chi$. Because of the random unitary nature, there are no systematic patterns in $p(z,t)$: $\pdiagz=1/D$. Therefore, in this case $F_d = F_c$, $\FXEBd=F_\text{XEB}$ and $\Fidd=F_\text{id}$. Furthermore, we will also consider $F_e  = F_\text{XEB}/F_\text{id}$.
%Our estimators $F_d$ and $F_e$, in our case where $\pdiagz$ is constant, are given by $F_d = 2(F_\text{XEB}+1)/(F_\text{id}+1) - 1$ and $F_e = F_\text{XEB}/F_\text{id}$. 
Our main result is:
\begin{align}
    &\Ebb_{W,\mathbf{U}}\L[F_\text{XEB}+1\R] \approx  \Ebb_{\mathbf{U}}\L[F_\text{id} +1\R] \L(F + \exp[-\frac{d-1}{d} \frac{N_A}{\chi}]\frac{1-F}{2}\R)~,\label{eq:random_MPS_prediction}
\end{align}
where $\Ebb_{W,\mathbf{U}}$ denotes averaging over the random unitaries $W$ and $\mathbf{U} = \{U_j\}$, and $\Ebb_{\mathbf{U}}$ denotes averaging over $\mathbf{U}$.

The second term describes a crossover between two regimes: $\chi \ll N_A$ and $\chi \gg N_A$. In the former regime, $\exp[-\frac{d-1}{d} \frac{N_A}{\chi}] \approx 0$, and $F_\text{XEB}+1 \approx \L(F_\text{id} +1\R) F $. Furthermore since $\chi$ is small, both $F_\text{id}$ and $F_\text{XEB}$ are exponentially large and
\begin{equation}
    F \approx \frac{F_\text{XEB}+1}{F_\text{id} +1} \approx \frac{F_\text{XEB}}{F_\text{id}} = F_e~.
    \label{eq:Frel}
\end{equation}

In the latter regime, $\exp[-\frac{d-1}{d} \frac{N_A}{\chi}] \approx 1$, and $F_\text{XEB}+1 \approx \L(F_\text{id} +1\R)\frac{1+F}{2}$, and 
\begin{equation}
    F \approx 2\frac{F_\text{XEB}+1}{F_\text{id}+1} - 1 = F_d~.
\end{equation}

Both $F_d$ and $F_e\equiv \L[\sum_z \pdiagz \tilde{q}(z)\tilde{p}(z) - 1\R]/\L[\sum_z \pdiagz \tilde{p}(z)^2 - 1\R]$ accurately estimate $F$ at late times. $F_d$ is valid since $\chi \approx D \gg N_A$. However $F_e$ is also valid since $F_\text{id} \approx 1$ and so $F_\text{XEB} \approx F$. These results give the regime of validity of $F_d$ and $F_e$ sketched in Fig.~\ref{fig:random_MPS}(d). Further properties of $F_e$ are discussed in Sec.~\ref{sec:Fe_summary}.

\subsection{Remarks}
The random MPS model qualitatively captures the short time physics. However, it is not quantitatively accurate, at least to first order in $D^{-1}$. The simplest example of this is that the state norm has non-zero variance: $\text{var}_\mathbf{U}\L[ \abs{\braket{\Psi_\mathbf{U}}}^2\R] \neq 0$. As another example, $F_\text{id}$ is, on average, bigger than 1, while $\Ebb_t \L[\Fidd \R] \leq 1$ for Hamiltonian dynamics [Eq.~\eqref{eq:aveFidd}]. Furthermore, the randomness in the MPS does not capture the systematic trends in finite temperature states that lead to non-trivial $\pdiagz$ factors. Therefore, the random MPS model cannot be a qualitative model for non-equilibrium quench-evolved states at either finite or infinite temperature. Nevertheless, the qualitative conclusions of the random MPS model --- in particular the regimes of validity of $F_d$ and $F_e$ --- are consistent with numerical results at finite temperature.

\subsection{Derivation of result}
We model the error operator $V(\tau) = W e^{i\theta X} W^\dagger$, where $W$ is a Haar-random unitary supported on $A$~[Fig.~\ref{fig:random_MPS}(b)]. This models the effect of scrambling and makes the calculation tractable. We first integrate over $W$ (a unitary two-design calculation~\cite{gross2007evenly,collins2016random}), and obtain two terms corresponding to the permutations of two elements: $\sigma_W = \id$ and $\sigma_W = \swap$. We then integrate over $U$ to obtain effective Ising chains for each value of $\sigma_W$: the integration over unitaries $U_j$ leaves a degree of freedom on each sites with two possible states, again corresponding to (local) $\id$ and $\swap$ permutations~[Fig.~\ref{fig:random_MPS}(c)]. These degrees of freedom ``interact'' with their nearest neighbors ferromagnetically: the lowest energy configurations are those with all sites in the $\id$ state or all sites in the $\swap$ state. This ferromagnetic interaction is encoded in the transfer matrix (of the partition function) $T_0$~\eqref{eq:transfer_matrices}. 

The presence of the error operator $V(\tau)$ modifies this effective Ising chain when $\sigma_W = \swap$: it introduces an effective magnetic field on the region $A$ supporting the error.

Our quantities of interest can be expressed in terms of the partition functions of these two Ising chains. We call these partition functions $f_{X,\id}$ and $f_{X,\swap}$, corresponding to the value of $\sigma_W$. For example, the average value of $F_\text{id}+1$ over all possible choices of $U_j$ is exactly the partition function $f_{X,\id}$. The presence of an error changes the value of $F_\text{XEB}+1$. $F_\text{XEB}+1$, again averaged over unitaries $U_j$ and $W$, can be expressed in terms of $f_{X,\id}$, $f_{X,\swap}$, and the (averaged) fidelity $F$.

Our calculation gives the expressions:
\begin{align}
     &\Ebb_{W,\mathbf{U}}\L[F\R] = \frac{\abs{\tr V}^2-1}{D_A^2-1} \Ebb_{\mathbf{U}}\L[\abs{\braket{\Psi_\mathbf{U}}{\Psi_\mathbf{U}}}^2\R]+ \frac{\abs{D_A^2 -\tr V}^2}{D_A\L(D_A^2-1\R)} \Ebb_{\mathbf{U}}\L[\tr(\rho_A^2)\R]\approx \frac{\abs{\tr V}^2-1}{D_A^2-1} \equiv \bar{F}~,\\
    &\Ebb_{W,\mathbf{U}}\L[F_\text{XEB}+1\R]
    = f_{X,\id} \bar{F} + \frac{f_{X,\swap}}{D_A} \L(1-\bar{F}\R)~,\\
    &\mathbb{E}_\mathbf{U} \L[F_\text{id} +1\R] = f_{X,\id} ~,
\end{align}
where the coefficients $f_{X,\id}$ and $f_{X,\swap}$ are the partition functions of the Ising spin models, defined by the transfer matrices:
\begin{align}
    T_0 &= \frac{1}{\chi^2 d^2 - 1}\begin{pmatrix}
    \chi^2d^2-d & \chi d(d-1)\\
    \chi d(d-1) & \chi^2d^2-d 
    \end{pmatrix}~,~T_1 = \begin{pmatrix}
    1 & 0 \\
    0 & \sqrt{d}
    \end{pmatrix} T_0 \begin{pmatrix}
    1 & 0 \\
    0 & \sqrt{d}
    \end{pmatrix}~. \label{eq:transfer_matrices}
\end{align}
These transfer matrices are obtained from the Weingarten coefficients that arise from the integration over random unitaries~\cite{collins2016random}. The partition functions are
\begin{align}
    f_{X,\id} &\equiv \frac{\chi}{\chi+1}\frac{\chi d + d}{\chi d + 1} \begin{pmatrix}
    1 & 1
    \end{pmatrix}  T_0^{N-1} \begin{pmatrix}
    1 \\
    1
    \end{pmatrix} = 2 \frac{\chi}{\chi+1} \L(\frac{\chi d + d}{\chi d + 1} \R)^N~,\\
    f_{X,\swap} &\equiv \frac{\chi}{\chi+1}\frac{\chi d + d}{\chi d + 1}\begin{pmatrix}
    1 & 1
    \end{pmatrix}  T_0^{k-1}\begin{pmatrix}
    1 & 0 \\
    0 & \sqrt{d}
    \end{pmatrix} T_1^{N_A}\begin{pmatrix}
    1 & 0 \\
    0 & \sqrt{d}
    \end{pmatrix} T_0^{N-N_A-k} \begin{pmatrix}
    1 \\
    1
    \end{pmatrix}~.
\end{align}
As a reminder, $f_{X,\id}$ is the partition function of the ferromagnetic effective Ising model, while $f_{X,\swap}$ is the partition function of the same Ising spin chain with an additional (very strong) effective magnetic field applied on region $A$, that effectively locks all the Ising spins in $A$ to the swap state $\mathcal{S}$.

We can simplify $f_{X,\id}$ and $f_{X,\swap}$ by noting that the dominant contributions to $f_{X,\id}$ are the configurations with all sites in the $\id$ state, and all sites in the $\swap$ state. Meanwhile, the dominant contribution to $f_{X,\swap}$ is the configuration with all sites in the $\swap$ state. Furthermore, $T_0$ has eigenvalues  $\frac{\chi d -d}{\chi d - 1}$ and $\frac{\chi d +d}{\chi d +1 }$, while $T_1$ has eigenvalues $\frac{(\chi d)^2 - d^2}{(\chi d)^2 - 1}$ and $d$. Therefore, for large $N_A$, $\frac{f_{X,\swap}}{D_A f_{X,\id}}$ is proportional to the ratio of their largest eigenvalues:
\begin{equation}
    \frac{1}{D_A}\frac{f_{X,\swap}}{f_{X,\id}} \approx \frac{1}{2}\L(\frac{\chi d+ 1}{\chi d + d}\R)^{N_A}\\
    \approx \frac{1}{2} \exp(- \frac{d-1}{d} \frac{N_A}{\chi})~.
\end{equation}
where the factor of a half comes from the breaking of the global Ising $\id \leftrightarrow \swap$ symmetry.

\begin{comment}
More precisely,
\begin{align}
    \frac{f_{X,\swap}}{D_A} &= \frac{\chi}{\chi+1} \L( \frac{\chi d + d}{\chi d + 1}\R)^{N-N_A} \Bigg[\frac{(\chi d + 1)^2}{\chi^2 d^2 + d} 1^{N_A-1} + \frac{(\chi+1)^2}{\chi^2 d + 1} \L( \frac{\chi^2 d -1}{\chi^2 d^2  -1}\R)^{N_A-1}\Bigg]\\
    &\approx \frac{\chi}{\chi+1} \L( \frac{\chi d + d}{\chi d + 1}\R)^{N-N_A} \L[1 + \frac{1}{d^{N_A}}\R]~.
\end{align}
\end{comment}
This gives our main result Eq.~\eqref{eq:random_MPS_prediction}, which predicts the relationship between the cross-entropy benchmark $F_\text{XEB}$, fidelity, and the parameters in our model: the operator size $N_A$ and MPS bond dimension $\chi$.

\section{Details of numerical simulations}
\label{app:simulation_details}
\begin{comment}
\begin{table*}
    \centering
    \begin{tabular}{c  cccc} \toprule
        \textbf{Model} & \textbf{Hamiltonian} & \textbf{Initial State} & \textbf{Error} & \textbf{Parameter values}\\
        \midrule
         1D Bose Hubbard &   & $\ket{1,1,\cdots,1}$ & $\{n_j\}$ & $(J,U) = (1,0.5)$\\
        1D Fermi Hubbard& & $(t,U)=(1,1)$\\
1D Trapped ion& & $\{ \sigma^x_j, \sigma^z_j\}$ & $(J_0, \tilde{h}_z,\alpha) = (1,0.7,1)$\\
2D Rydberg& & $\ket{0,0,\cdots, 0}$ & $\{ \sigma^z_j\}$ & $(\Omega,\delta)=(1,0.7)$\\
\bottomrule
    \end{tabular}
    \caption{Details of the four models studied: The Bose and Fermi Hubbard models in 1D, a long-ranged spin-1/2 model applicable to trapped ion simulators, and a 2D Rydberg blockaded model. Specifically, we list the Hamiltonian, initial state, error model, and parameter values presented in Fig.~\ref{fig:Fig1}.}
    \label{tab:details}
\end{table*}
\end{comment}
Below we summarize the details of the models simulated: their Hamiltonians, initial states, error models, and Hamiltonian parameter values.
\subsection{Bose-Hubbard model}
We numerically simulate quench evolution under the one-dimensional (1D) Bose-Hubbard Hamiltonian~\cite{jaksch1998cold} in the presence of noise, which is often dominated by the photon scattering from optical lattices.
In the quantum trajectory formalism~\cite{daley2014quantum}, such noisy dynamics is described by stochastically applying the particle number operator $\{n_j\}$ at a rate $\gamma$ for each site. An ensemble of such pure state evolution represents a mixed state $\rho(t)$  in open dynamics. The 1D Bose-Hubbard Hamiltonian is:
\begin{equation}
    H_\text{BH} = -\Omega \sum_{j=1}^{N-1} \left(b^\dagger_j b^{\mathstrut}_{j+1} + \text{h.c.} \right) + \frac{U}{2} \sum_{j=1}^N n_j (n_j -1)~,
    \label{eq:BH_Ham}
\end{equation}
where $b^\dagger_j$ is the bosonic creation operator at site $j$, and $n_j = b^\dagger_j b^{\mathstrut}_j$ is the bosonic number operator at site $j$.
In Fig.~1 and 2, we use the parameters  $(\Omega,U) = (1,2.87)$.
To increase the ergodicity of our quench, in Fig.~3,4, we use the parameters $(\Omega,U) = (1,0.5)$, whose ground state is in the superfluid phase. We take as initial state the ``unity filling state" $\ket{1,1,\dots,1}$, which is the ground state deep in the Mott insulating phase $U \gg \Omega$.

Our error operators are $\{n_j\}$, which model the effects of scattering from optical lattice photons~\cite{pichler2010nonequilibrium}. In Fig.~1(b), we use an overall error rate of $\gamma = 0.05$.

 In Fig.~2(a), we show data for a system of 6 bosons on 6 sites, and in Fig.~1(b) and Fig.~2(b,c),
we show data for a system with 10 bosons on 10 sites (with total Hilbert space dimension 92378).

\subsubsection{Disorder Averaging}
In Fig.~3(b-d), we present disorder averaged results for the response curve and the sample complexity of $F_d$ against time. Disorder averaging is done to remove the large fluctuations present at short times in order to clearly observe systematic behavior.

For simplicity, we choose the following ensemble of Bose-Hubbard models with disordered hopping strengths:
\begin{equation}
    H_\text{BH} = - \sum_{j=1}^{N-1} \Omega_j \left(b^\dagger_j b^{\mathstrut}_{j+1} + \text{h.c.} \right) + \frac{U}{2} \sum_{j=1}^N n_j (n_j -1)~,
    \label{eq:BH_Ham_Disorder}
\end{equation}
with $\Omega_j = 0.2 m_j$, for $m_j$ a random integer from 1 to 10. As in  Eq.~\eqref{eq:BH_Ham_Disorder}, we choose $U = 0.5$. We chose to disorder the hopping strengths in order to observe different dynamics at short terms. If, for example, we had chosen a disordered potential $\Delta_j n_j$ instead, due to the homegeneous initial state $\ket{1,\dots,1}$, the resultant states $\ket{\psi(t)}$ would not be significantly different from each other across disorder realizations. Lastly, an error at $n_{\lfloor L/2 \rfloor}$ is applied at a fixed, early time $t=5$. This is because the error operator $n_j$ acts trivially on the initial state $\ket{1,\dots,1}$.

\subsection{Fermi-Hubbard model}
The 1D Fermi-Hubbard model Hamiltonian is:
\begin{equation}
    H_\text{FH} =  \sum_{j=1}^N \L[-\Omega\sum_{\sigma = \uparrow,\downarrow}  \L(c^\dagger_{j,\sigma} c^{\mathstrut}_{j+1,\sigma} + \text{h.c.}\R) + U n_{j,\uparrow} n_{j,\downarrow}\R]~,
\end{equation}
where $c^\dagger_{j,\sigma}$ is the creation operator of a fermion at site $j$, with spin $\sigma \in \{\uparrow, \downarrow\}$, and $n_{j,\sigma} = c^\dagger_{j,\sigma}c^{\mathstrut}_{j,\sigma} \in \{0,1\}$ is the fermion occupation number. In Fig.~2 we use the parameter value $(\Omega,U) = (1,3)$, while in other figures we use$(\Omega,U) = (1,1)$. For any value of $U/\Omega$, this model is integrable by Bethe Ansatz~\cite{essler2005one}, but we numerically observe that it still passes the $2$nd no-resonance condition, after disregarding degeneracies [Eq.~\eqref{eq:k_no_resonance}]. As discussed in Section~\ref{app:Hamil_k_design}, degeneracies can be disregarded for the purposes of Theorem~\ref{thm:bitstring_PT}.

We take as initial state the anti-ferromagnetic state at half filling, $\ket{\uparrow,\downarrow, \dots}$, and as error operators the total fermion occupation $\{n_{j,\uparrow}+n_{j,\downarrow}\}$ at each site. As with the Bose-Hubbard model, this models scattering from optical lattice photons~\cite{sarkar2014light}. In Fig.~1(c), we use a total error rate of $\gamma = 0.2$, in a system with 5 up and 5 down fermions on 10 sites (with total Hilbert space dimension 63504).

\subsection{Trapped ion model}
Our trapped ion model has Hamiltonian:
\begin{equation}
    H_\text{TI} = \sum_{j=1}^{N-1}\sum_{i>j} \L(\Omega/\abs{i-j}^\alpha\R) \sigma^x_i \sigma^x_j + h_z \sum_{j=1}^N \sigma^z_j~,
    \label{eq:trapped_ion_model}
\end{equation}
with $\sigma^x_j$ and $\sigma^z_j$ the Pauli $x$- and $z$-operators on site $j$. $H_\text{TI}$ is adapted from the experiment of Ref.~\cite{zhang2017observation}. Due to the long-ranged interactions, following Ref.~\cite{zhang2017observation} we fix the ``rescaled" parameter value $\tilde{h}_z =  h_z/(N^{-1} \sum_{j = 1}^{N-1} \sum_{i>j} 1/\abs{i-j}^{\alpha})$ to ensure consistent behavior across system sizes, where $\alpha$ is the exponent of the power-law interactions. As in Ref.~\cite{zhang2017observation}, our initial state is the all-downs computational basis state: $\ket{\downarrow,\dots,\downarrow}$. We use the rescaled parameters $(\Omega, \tilde{h}_z,\alpha) = (1,0.7,1)$. In Fig.~1(d), we use on-site spin flip and dephasing errors: $\{\sigma^x_j, \sigma^z_j\}$, with a total error rate of $\gamma = 0.05$ in a system of 16 qubits (with a total Hilbert space dimension of $2^{16} = 65536$).

In Fig. 4(c) we learn the disorder strengths $\Delta_j$ of a disorder term $H_\text{dis} = \sum_j \Delta_j \sigma^z_j$.

\subsection{Rydberg blockaded model}
We simulate a Rydberg blockaded model in two dimensions (2D), as recently experimentally realized in Ref.~\cite{bluvstein2021controlling}. Here each site can be either in the ground $\ket{\downarrow}$ or Rydberg $\ket{\uparrow}$ state, subject to the constraint that there are no two adjacent Rydberg excitations $\ket{\uparrow}$. Its Hamiltonian is:
\begin{equation}
    H_\text{PXP} = \Omega \sum_{j=1}^N PX_j + \delta \sum_j n_j~,
\end{equation}
with $PX_j \equiv \L[\bigotimes_{\langle i, j\rangle} \ketbra{\downarrow}{\downarrow}_i\R] \sigma^x_j$ the projected spin flip operator --- a spin flip on site $j$ is allowed only if all the nearest neighbors $i$ of site $j$ are in state $\ket{\downarrow}$.

We use the all-downs state $\ket{\downarrow,\dots,\downarrow}$ as our initial state. At zero detuning $\delta=0$, this would be an infinite temperature state ($\langle H \rangle = \tr(H)/D$). In Fig. 2 we use the parameter values $(\Omega,\delta) = (1,1.5)$ and in all other figures we use the parameters $(\Omega,\delta) = (1,0.7)$ to tune $\ket{\downarrow,\dots,\downarrow}$ away from infinite temperature. Furthermore, when $\delta=0$, the Hamiltonian is particle-hole symmetric, leading to a different kind of universal behavior, detailed in Section~\ref{app:PHS}. In Fig.~1(e), we use dephasing error $\{\sigma^z_j\}$ at a total error rate of $\gamma = 0.05$, in a $5\times 5$ square grid of Rydberg atoms (with total Hilbert space dimension of 55447).

\section{Limitations of benchmarking protocol: Failure cases}
\label{app:failure_cases}
As discussed in the main text, our benchmarking formula can fail when: 1) the $k$-th no resonance condition is not satisfied, 2) the effective dimension $D_\beta$ is small, 3) the approach to equilibrium is slow, 4) the error is highly non-local, or 5) the readout basis states have well-defined symmetry quantum numbers. Below, we detail our findings for each case, and discuss alternate benchmarking formulae.

\subsection{Violating the no-resonance condition: Particle-hole symmetry} \label{app:PHS}
The second no-resonance condition~\eqref{eq:k_no_resonance} is extensively violated in the presence of a particle-hole symmetry (PHS). By this we mean a reflection symmetry in the spectrum: for every eigenvalue $E$, there is one at $-E$. For any given $E,E'$, the equation $E_1 - E_2 + E_3 - E_4 = 0$ not only has solutions $(E_1, E_2, E_3, E_4) = (E, E, E', E')$ and $(E, E', E', E)$, but also has solutions: $(E_1, E_2, E_3, E_4) = (E,- E, E', -E'),~(E,-E, -E', E')$. This will give rise to additional terms in the second Hamiltonian twirling identity [Eq.~\eqref{eq:Ham_2_design}]:
\begin{align}
    &\Ebb_t \L[ U_t^{\otimes 2} A U_{-t}^{\otimes 2}\R] \\
    &\approx \sum_{E,E'} \bigg[ \ketbra{E,E'}{E,E'} A \ketbra{E,E'}{E,E'} + \ketbra{E,E'}{E,E'} A \ketbra{E',E}{E',E} + \ketbra{E,-E}{E,-E} A \ketbra{E', -E'}{E',-E'} \bigg]~,\nonumber
\end{align}
suppressing the subleading collision terms. This can arise from, an \textit{anti-unitary} charge conjugation operator, for example. Here, we focus on a simpler example. PHS is present in the PXP model with no detuning ($\delta = 0$): $H_\text{PXP} = \sum_j PX_j$. $H_\text{PXP}$ anticommutes with the unitary operator $\mathcal{C} \equiv \prod_j \sigma_j^z$: $\mathcal{C} H_\text{PXP}\mathcal{C} =  - H_\text{PXP}$~\cite{turner2018quantum}. This symmetry is present in arbitrary dimensions--- we will present numerical evidence for the two dimensional PXP model.

The PXP model is particularly convenient, for a number of reasons:
\begin{itemize}
    \item[(i)] The computational basis states $\ket{z}$
are eigenstates of $\mathcal{C}$: $\mathcal{C}\ket{z} = \pm \ket{z}$. This implies that easily-prepared bitstring initial states such as $\ket{\Psi_0} = \ket{\downarrow\dots \downarrow}$ are also eigenstates of $\mathcal{C}$.
\item[(ii)] $H_\text{PXP}$ is real and thus possesses a time-reversal symmetry: $H_\text{PXP} = H^*_\text{PXP}$. In particular, its eigenstates are real in the computational basis. 
\end{itemize}

The breakdown of the second no-resonance condition can be seen when we compute the time average of $P^{(2)}(t) = \sum_z \pdiagz \tilde{p}(z,t)^2$:
\begin{equation}
\Ebb_t \L[F_\text{id,d}(t) \R] \approx \sum_z \bigg[ 2\bigg(\sum_E\abs{\braket{z}{E}\braket{E}{\Psi_0}}^2\bigg)^2 + \big\vert\sum_E\braket{z}{E}\braket{z}{-E} \braket{E}{\Psi_0} \braket{-E}{\Psi_0}\big\vert^2\bigg]/\pdiagz~.
\end{equation}
Using properties (i) and (ii) above, $\sum_E\braket{z}{E}\braket{z}{-E} \braket{E}{\Psi_0} \braket{-E}{\Psi_0} = \pdiagz$ and we obtain $\Ebb_t \L[F_\text{id,d}(t) \R] \approx 3$ instead of the usual value of two. This is illustrated in Fig.~\ref{fig:failure_cases1}(a). We note that this result is contingent on conditions (i) and (ii) above --- if these conditions were not satisfied, $\Ebb_t \L[F_\text{id,d}(t) \R]$ will have a non-universal value between two and three. For example, if we use an initial state that is \textit{not} particle-hole symmetric, such as a low-energy Boltzmann state, $\Ebb_t \L[F_\text{id,d}(t) \R] \approx 2$, and $F_d$ remains valid.

In particular, this implies that our $F_d$ formula will not estimate the fidelity. Applying our simplifying assumptions in Eq.~\eqref{eq:tildeq_approx}, we obtain: $F_d \approx \frac{4}{3}F - \frac{1}{3}$. Instead, a modified formula should be employed:
\begin{equation}
    F_d^\text{(PH)} = \frac{3}{2} \frac{\sum_z q(z) \tilde{p}(z)}{\sum_z p(z) \tilde{p}(z)} - \frac{1}{2}~. \label{eq:real_PT_Fd}
\end{equation}
However, a similar analysis reveals that $\Freld$ can be used with no modification. This is illustrated in Fig.~\ref{fig:failure_cases1}(b).

\subsubsection{Real Porter-Thomas distribution}
The value of three above is the second moment of the \textit{real Porter-Thomas distribution}~\cite{porter1956fluctuations} (stated, for simplicity, in the $D\rightarrow \infty$ limit):
\begin{equation}
    P_{\text{PT},\mathbb{R}}(x) = \frac{\exp(-x/2)}{\sqrt{2\pi x}}~.
    \label{eq:real_PT}
\end{equation}
This distribution can be viewed as the overlap probability of a \textit{real-valued} random vector $\ket{\Psi}$. It has $k$-th moments: $\int_0^\infty dx~x^k P_{\text{PT},\mathbb{R}}(x) = (2k-1)!! \equiv (2k-1)(2k-3)\cdots 3\cdot 1$.

Intuitively, if both the initial state $\ket{\Psi_0}$ and $\ket{z}$ are even under $\mathcal{C}$, the overlap will be purely real: $\braket{z}{\Psi(t)} = 2\sum_{E\geq 0} \braket{z}{E}\!\braket{E}{\Psi} \cos(E t)$. We intuitively expect this to be a random variable and for $p(z,t) = \abs{\braket{z}{\Psi(t)}}^2$ to follow the real Porter-Thomas distribution. This also holds if $\ket{\Psi_0}$ or $\ket{z}$ are odd under $\mathcal{C}$.

We can prove that $\tilde{p}(z,t)$ (both with fixed time, and fixed configuration $z$) follows $P_{\text{PT},\mathbb{R}}$, assuming conditions (i) and (ii) above, and a \textit{particle-hole symmetric} $k$ no-resonance condition. The proof is identical to that of Theorem~\ref{thm:bitstring_PT}, except that instead only considering the permutations of $k$ elements, we consider \textit{perfect matchings} of $2k$ vertices. There are $(2k-1)!!$ such perfect matchings, giving the desired $k$-th moments. We emphasize that to satisfy both conditions (i) and (ii), both time-reversal and particle-hole symmetry must be present in the Hamiltonian, and the initial state and readout basis states must have definite PHS parity.

\subsection{Violating the no-resonance condition: Non-interacting models}
The second no-resonance condition is also violated in non-interacting models. In non-interacting models, energy eigenvalues are given in terms of single-particle eigenvalues. For concreteness, consider a non-interacting Bose-Hubbard model, Eq.~\eqref{eq:BH_Ham} with $U=0$. Let the eigenvalues of the single-particle Hamiltonian be $\epsilon_k$. An $N$-particle eigenstate is labelled by the occupation $\{n_k\}$ of the single-particle modes $k$, with energy $\sum_k n_k \epsilon_k$. This violates the 2 no-resonance condition: For any $i,j,k,l$, there are eigenvalues $E_1 = \epsilon_{i} + \epsilon_{j}, E_2 = \epsilon_{k} + \epsilon_{l}$ and $E_3 = \epsilon_{i} + \epsilon_{k}, E_4 = \epsilon_{j} + \epsilon_{l}$ such that $E_1 + E_2 = E_3 + E_4$. Therefore, our Theorem does not apply to non-interacting models, reflected in Fig.~\ref{fig:failure_cases1}(c). Consequently, neither $F_d$ or $F_e$ estimate the many-body fidelity (Fig.~\ref{fig:failure_cases1}(d)). However, integrable and non-integrable dynamics may be quantitatively similar at short times, at which times the integrability breaking terms have not had sufficient time to appreciably change the dynamics. Since $F_d$ is designed to work for these short quenches, $F_d$ appears to estimate $F$ in this case, as seen in Fig.~\ref{fig:failure_cases1}(d).

\begin{figure}[t]
    \centering
    \includegraphics[width=0.95\textwidth]{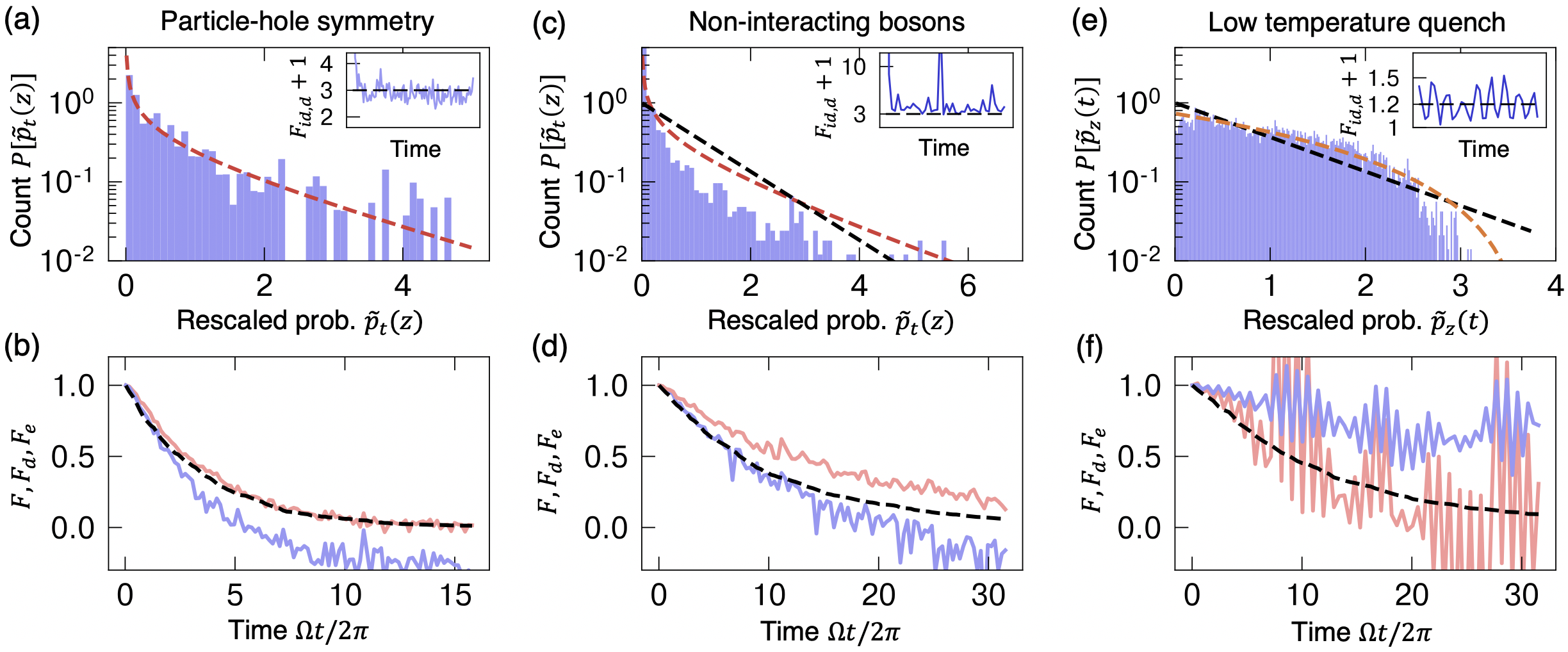}
    \caption{Failure cases of benchmarking protocol: Violation of no-resonance conditions and low effective dimension. (a) The particle-hole symmetry in the 2d PXP model violates the $k$ no-resonance conditions~\eqref{eq:k_no_resonance}. Under special conditions (satisfied here in the 2d PXP model quenched from $\ket{0\cdots 0}$), $\tilde{p}(z,t)$ follows the \textit{real} Porter-Thomas distribution, and $\Fidd +1$ has an equilibrium value of 3 (inset). (b) The $F_d$ formula (blue) breaks down, and must be modified as in Eq.~\eqref{eq:real_PT_Fd}. In contrast, $F_e$ (red) can be used without modification. (c) Non-interacting models also violate the no-resonance conditions. Here, we simulate a quench of the Bose-Hubbard model with no interaction $U=0$. The $\tilde{p}(z,t)$ follows neither PT distributions (red and black), and $\Fidd +1$ has large temporal fluctuations (inset). (d) As a consequence, neither $F_d$ nor $F_e$ approximate $F$, although $F_d$ is a good estimate at short times. (e) A low-temperature quench has a small effective dimension. We realize this with a Gibbs initial state $\ket{\Psi_0} = \sum_E \exp(-\beta/2 E) \ket{E}$, with $\beta = 4$ in a Bose-Hubbard model with 8 particles on 8 sites. $\tilde{p}_z(t)$, for an arbitrarily chosen $z$, follows a Porter-Thomas distribution with a small dimension (orange): here $D_\beta(z) \approx 3.75$. (f) As a consequence, $F_d$ (blue) and $F_e$ (red) have large temporal fluctuations, and $F_d$ has a large systematic error. $F_e$ is less susceptible to a small $D_\beta$, and is on average close to $F$.}
    \label{fig:failure_cases1}
\end{figure}
\begin{figure}[t]
    \centering
    \includegraphics[width=0.95\textwidth]{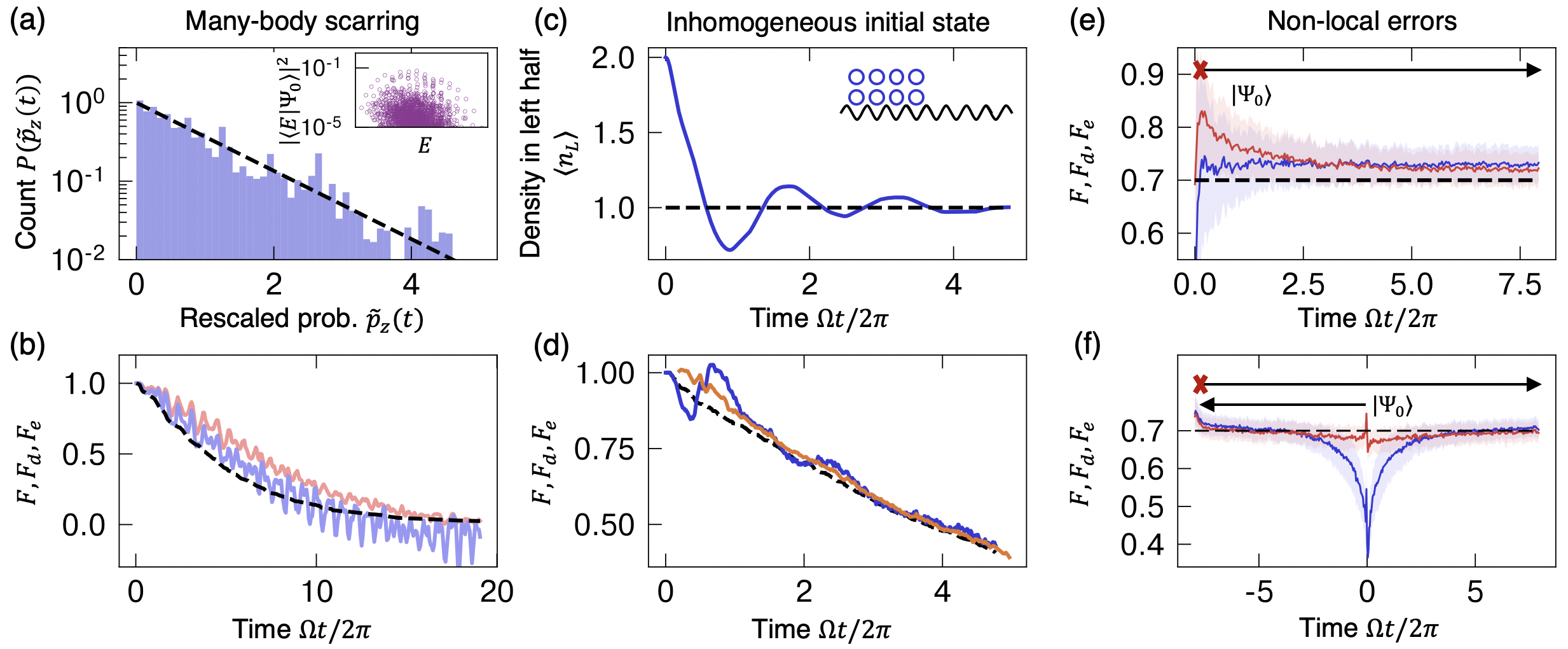}
    \caption{Failure cases of benchmarking protocol: Slow equilibration and non-local errors. (a) In the 1D PXP model, a quench of the Neel state $\ket{\downarrow \uparrow \downarrow \uparrow \dots}$ undergoes slow thermalization due to large overlap with a ``tower" of eigenstates (inset). $p_z(t)$ follows a PT distribution. (b) However, both $F_d$ (blue) and $F_e$ (red) show large oscillations and do not estimate $F$ well (c) If an inhomogeneous state is quench evolved, it undergoes slow hydrodynamic relaxation, evidenced by the decaying oscillations of the average density in the left half of a Bose-Hubbard chain (blue). (d) Under this quench, both $F_d$ (blue) and $F_e$ (not plotted) exhibit oscillations. A modified $F_e$ (orange) using a moving-averaged $p_\text{avg}$ yields a closer estimate of the fidelity. (e) $F_d$ and $F_e$ response to an initial error. At short times $F_d$ (blue) is a better estimate of $F$ than $F_e$ (red). Solid lines are average values over an ensemble of disordered Hamiltonians, with shaded areas indicating variance. (f) We apply a local error on a \textit{backward time-evolved} state $\ket{\Psi(-\tau)}$ to engineer a non-local error on the initial state. In this case $F_e$ (red) estimates $F$, while $F_d$ (blue) does not, consistent with Appendix.~\ref{app:random_MPS}.}
    \label{fig:failure_cases2}
\end{figure}

\subsection{Small effective dimension: Low temperature quench}
Our protocol also breaks down when the effective dimension $D_\beta$ is small. This happens in low temperature quenches, for example. In the extreme case that $\ket{\Psi_0}$ is an eigenstate, $D_\beta = 1$. When $\ket{\Psi_0}$ is a low temperature state, it's dynamics are effectively constrained to a low-energy subspace. Following Theorem~\ref{thm:bitstring_PT}, in Fig.~\ref{fig:failure_cases1}(e), the variable $\tilde{p}_z(t)$ approximately follows a Porter-Thomas distribution~\eqref{eq:PorterThomas} with the effective dimension \textit{of the outcome} $z$, $D_\beta(z)$. Meanwhile, the distributions $\tilde{p}_t(z)$ show large fluctuations over time, again because of the small effective dimension.

 Our protocol's systematic error and fluctuations scale inversely with $D_\beta$. Therefore, it is less applicable to low temperature quenches. In Fig.~\ref{fig:failure_cases1}(f), we see a large discrepancy between $F_d$ and $F$ and large fluctuations over time. However, since the systematic error of $F_e$ is less senstive to $D_\beta$ [its $O(D_\beta^{-1})$ term vanishes], $F_e$ estimates $F$, when averaged over a time window.

\subsection{Slow equilibration: Quantum many-body scar states}
Some many-body systems exhibit long-lived oscillations which delays equilibration of the quenched state. In the presence of these oscillations, normalizing the measurement outcomes $p(z)$ with the \textit{equilibrium} distribution $\pdiagz$ gives incorrect estimators.

One prominent example of long-lived oscillations is the recent discovery of \textit{quantum many-body scar states}~\cite{bernien2017probing,turner2018quantum}. Originally experimentally discovered in a Rydberg quantum simulator, this phenomenon can be captured with the 1D PXP model:
\begin{align}
    H_\text{1D-PXP} = \Omega \sum_{j=2}^{L-1} P_{j-1}^\downarrow X_j P_{j+1}^\downarrow + X_0 P_1^\downarrow + P_{L-1}^\downarrow X_L + \delta \sum_j n_j~,
\end{align}
where $P^\downarrow_j \equiv \ketbra{\downarrow}{\downarrow}_j$ is the projector onto the atomic ground state and $n_j \equiv \ketbra{\uparrow}{\uparrow}_j$ is the projector on the Rydberg excited state. A quench of the Neel state $\ket{\downarrow \uparrow \downarrow \uparrow \cdots \downarrow \uparrow}$ was observed to exhibit long-lived oscillations --- in particular, persistent revivals of observables such as the magnetization and many-body fidelity. These are attributed to a tower of eigenstates, energy eigenstates with which $\ket{\Psi_0}$ has considerable overlap [Fig.~\ref{fig:failure_cases2}(a), inset]. Due to special properties of these eigenstates --- most notably low half-chain entanglement entropy --- these states are dubbed \textit{many-body scar states}. While the phenomenon is most apparent with zero detuning, it persists with non-zero detuning. In our results in Fig.~\ref{fig:failure_cases2}(a,b), we use a non-zero detuning $\delta = 0.2 \Omega$ to break the particle-hole symmetry.

These revivals eventually decay in time. Therefore, $\tilde{p}_z(t)$ follows a PT distribution at late times [Fig.~\ref{fig:failure_cases2}(a)]. However, these long-lived oscillations are reflected in $F_d$ and $F_e$ [Fig.~\ref{fig:failure_cases2}(b)] --- $F_d$ is a closer estimate to $F$ than $F_e$.

\subsection{Slow equilibration: Inhomogeneous initial state}

If our initial state is strongly inhomogeneous, it will exhibit hydrodynamic oscillations before equilibrating~\cite{vasseur2016nonequilibrium}. At short times, the outcome probabilities $p(z,t)$ will be far from the relatively homogeneous $\pdiagz$. We illustrate this with the initial state $\ket{\Psi_0} = \ket{2,2,2,2,0,0,0,0}$ on a Bose-Hubbard model with 8 bosons on 8 sites. We measure its relaxation with the average density in the left half $\langle n_L \rangle \equiv \langle n_1 + n_2 + n_3 + n_4\rangle/4$, which exhibits decaying oscillations in Fig.~\ref{fig:failure_cases2}(c). These oscillations are reflected in $F_d$ and $\Freld$, which show large deviations from $F$ (Fig.~\ref{fig:failure_cases2}(d)). We obtain better performance by normalizing $p(z,t)$ and $q(z,t)$ by a moving average $p'_\text{avg}(z) \equiv \int_{t-\Delta t/2}^{t+\Delta t/2} dt~p(z,t)$, for a short interval $\Delta t$. We find that the modified $\Freld$, $F'_e(t) \equiv (\sum_z q(z,t)\tilde{p}(z,t) - 1)/(\sum_z p(z,t) \tilde{p}(z,t) - 1)$ is a closer approximation to $F$.

\subsection{Non-local error on low entanglement states}
In Appendix~\ref{app:random_MPS}, the random MPS model predicts that when the error size $N_A$ is much smaller than the bond dimension $\chi$, $F_d$ is a better approximation to $F$ at short times than $F_e$. This condition is satisfied in typical open system dynamics, since $\chi$ grows exponentially in time will $N_A$ only grows linearly in time.

To verify our claim, we apply a single error at the start of the quench ($t=0$) and compute $F_d$ and $F_e$. Averaging over a large number of Bose-Hubbard Hamiltonians with disordered hopping strengths, we verify that $F_d$ on average estimates $F$ at very short times (albeit with large variance) (blue, Fig.~\ref{fig:failure_cases2}(e)), while $F_e$ (red) has a longer relaxation time.

To engineer a situation in the opposite limit, where $\chi \ll N_A$, we backwards time-evolve the initial state, then apply the local error on $\ket{\Psi(-\tau)}$. We then forward evolve the ideal and perturbed states. When $\abs{t}$ is small, this is equivalent to a non-local error $V(t+\tau)$ on the low-entanglement states $\ket{\Psi(t)}$. In this regime, as predicted, $F_d$ (blue, Fig.~\ref{fig:failure_cases2}(f)) deviates from $F$ --- reaching a minimum value of $2F-1$, consistent with Eq.~\eqref{eq:random_MPS_prediction} --- while $F_e$ is close to $F$ (red).

We emphasize that such a situation is unlikely to arise in practical situations, and therefore $F_d$ is more appropriate for short time quench dynamics than $F_e$. We illustrate this in an example of target state benchmarking in Fig.~\ref{fig:target_state_benchmarking_local}(b).
 
\subsection{Non-ergodicity due to symmetry and readout in symmetry-resolved basis}
\label{sec:symmetry_nonergodicity}
We lastly discuss another case where $F_d$ does not estimate $F$ in the presence of symmetry and a readout basis that has well-defined symmetry quantum numbers.

In the presence of a symmetry, quantum states ergodically explore the Hilbert space of their symmetry sector. In generic readout bases, the presence of such symmetries does not impede $F_d$. However, if every state in the readout basis belongs to one symmetry sector, $F_d$ will not be able to detect relative phase errors, or phase decoherence, between the different symmetry sectors. Physically, we can understand this as non-ergodicity of certain information about the state: the phase information never propagates into the readout basis. Mathematically, this does not affect our Theorems~\ref{thm:temporal_PT} and~\ref{thm:bitstring_PT}. However, our ansatz that $q(z) = F p(z) + (1-F) p_\perp(z)$ is no longer valid. Instead, consider two states, expressed as superpositions of states in different symmetry sectors: $\ket{\Psi} = \sum_j c_j \ket{\Psi_j}$ and $\ket{\Psi'} = \sum_j d_j \ket{\Psi'_j}$. Then for a bitstring $z$ in sector $j$, the ansatz should be modified to (assuming the quench processes for $\ket{\Psi}$ and $\ket{\Psi'}$ are identical):
\begin{align}
q(z) = F_j p(z) + (1-F_j)p_\perp(z)~,~F_j \equiv \abs{\braket{\Psi_j}{\Psi'_j}}^2
\end{align}

Under this ansatz, $F_d$ does not reproduce the true fidelity $F = \abs{\sum_j c^*_j d_j \braket{\Psi_j}{\Psi'_j}}^2$.
Applying our ansatz, and recalling that $\pdiagz \equiv \Ebb[p(z,t)]$, $F_\text{XEB,d} \equiv \sum_z q(z) \tilde{p}(z) - 1$ takes the value $2\sum_j \abs{d_j}^2 F_j -1$

Consider the simplest case, $\ket{\Psi_j} = \ket{\Psi'_j}$. Then $F_j = 1$, and $F_\text{XEB,d} = 2\sum_j \abs{d_j}^2 -1 = 1$, regardless of the value of $c_j,d_j$. Therefore, $F_d$ is unable to estimate the fidelity. The infidelity arising from unequal populations $\abs{c_j}^2$ of the symmetry sectors may be captured by a suitably modified formula. However, infidelity arising from different phases of $c_j$ is intrinsically inaccessible from such measurements: if $d_j = \exp(i \theta_j) c_j$, then $q(z) = p(z)$ identically. 

Many analog simulators have readout bases that belong to a given symmetry sector. For example, itinerant particles in optical lattices readout configuration outcomes that have well-defined particle number, and their natural Hamiltonians also conserve particle number. However, states without well-defined particle number are not natural in such settings.

Such states that are in multiple symmetry sectors, however, may arise in other settings. Consider our trapped ion model $H_\text{TI}$~\eqref{eq:trapped_ion_model}. This Hamiltonian conserves parity $\prod_j \sigma^z_j$, and the computational basis states have well-defined parity. Furthermore, it is easy to prepare states without well-defined parity: for example, one starts with the all-downs state $\ket{\downarrow,\cdots,\downarrow}$, and applies a local rotation on a single site. Therefore, such a family of states cannot be benchmarked with the $F_d$ formula. As another example, such a scenario can arise in a spin chain, with a quench Hamiltonian that conserves magnetization $\sum_j \sigma_j^z$. Bitstrings in the computational basis have well-defined magnetization. In both cases, however, this can be remedied by measuring in the $X$-basis, in which outcomes do not have well-defined parity nor magnetization. This is explicitly demonstrated in a trapped-ion spin model in Fig.~\ref{fig:target_state_benchmarking_local}(e,f).
\begin{quote}

\end{quote}

 \section{Greedy algorithm for multiparameter estimation}\label{app:greedy_algorithm}
\subsection{Algorithm}
In this section we detail our greedy algorithm for simultaneous estimation of multiple Hamiltonian parameters, used in Fig.~4(c). Our algorithm is as follows. We are given a Hamiltonian $H(\pmb{\theta})$ with $n$ unknown parameters $\pmb{\theta}=\{\theta_j\}_{j=1}^n$. We want to determine the actual experimental parameters $\{\theta_j^*\}$, assuming that we can prepare a fiducial state $\ket{\Psi_0}$ with high fidelity.

\noindent Our procedure is as follows:
    \begin{quote}
\noindent 1. Evolve $\ket{\Psi_0}$ for a fixed time.
    
\noindent 2. Measure the configurations to approximate the distribution $q(z) = \bra{z} \rho(t) \ket{z}$.
    
\noindent 3. For some guess $\{\theta_j\}$, compute $p(z,\{\theta_j\}) = \abs{\bra{z} \exp(-iH(\{\theta_j\})t) \ket{\Psi_0}}^2$. Compute $F_d(\{J_j\})$.

\noindent 4. Choose a random permutation $\pi \in S_n$. For each $j$, maximize $F_d(\{\theta_j\})$ over $\theta_{\pi(j)}$, keeping all other parameters fixed.

\noindent 5. Repeat 4. for all $j=1,\cdots n$.
    
\noindent 6. Repeat 3-5. until $F_d$ converges. 
    
\noindent 7. Optionally repeat all steps with different initial guesses of $\{\theta_j\}$.
\end{quote}

This procedure is justified as long as $F_d$ does not vary too quickly in the range of $\{\theta_j\}$ considered. 

In Fig.~4(c), we use this procedure to determine an unknown disorder potential atop a known systematic Hamiltonian. Our greedy algorithm works well as long as the disorder potential is small, and is also robust to errors in the experimental evolution.
 
\section{Target state benchmarking}
In Fig. 4(a) of our main text, we demonstrate target state benchmarking to determine the unknown phase of a GHZ-like state in a 2D array of Rydberg atoms. We take the prepared state to have a phase $\phi^* = \pi/4$. We wish to measure $\phi^*$ by measuring the fidelity of the prepared state to target states with various values of $\phi$. Using our protocol, we estimate the fidelity by measuring $F_d$ after quench evolution.

In this section, we provide further demonstrations of target state benchmarking. We first consider the simplest setting of coherent local errors, which we demonstrate in the three other platforms studied. We then study benchmarking between states that are related by non-local deformations. We consider this setting for two applications: determining the Hamiltonian parameters corresponding to a ground state, as well as estimating the temperature of a low-energy state by measuring its overlap with the ground state.

\subsection{Coherent local errors}
\label{sec:target_state_1}
Here, we demonstrate target state benchmarking in the Bose-Hubbard, Fermi-Hubbard, and trapped-ion spin chain model, using the same Hamiltonian and readout basis discussed in Section~\ref{app:simulation_details}.

We consider locally rotated initial states: for the Bose-Hubbard model we consider a unity-filling Mott insulator state on 8 sites deformed by a local rotation $\phi$ on the middle two sites:
\begin{equation}
    \ket{\Psi(\phi)}_\text{BH} = \cos{\phi} \ket{1,1,1,1,1,1,1,1} + i \sin{\phi} \ket{1,1,1,2,0,1,1,1}~.
\end{equation}
For the Fermi-Hubbard model we consider the anti-ferromagnetic state with 8 particles on 8 sites, again rotated along the middle bond:
\begin{equation}
    \ket{\Psi(\phi)}_\text{FH} = \cos{\phi} \ket{\uparrow,\downarrow,\uparrow,\downarrow,\uparrow,\downarrow,\uparrow,\downarrow} + i \sin{\phi} \ket{\uparrow,\downarrow,\uparrow,d,0,\downarrow,\uparrow,\downarrow}~,
\end{equation}
where $d = \uparrow \downarrow$ denotes a doubly-occupied site, and $0$ denotes an empty site.
Lastly, for the trapped-ion model we consider the fully polarized state of 12 qubits, rotated by a pair of spin-flips across the middle bond
\begin{equation}
    \ket{\Psi(\phi)}_\text{TI} = \cos{\phi} \ket{\downarrow,\cdots,\downarrow,\downarrow,\downarrow,\cdots,\downarrow} + i \sin{\phi} \ket{\downarrow,\cdots,\downarrow,\uparrow,\downarrow,\cdots,\downarrow}~.
\end{equation}
A single bit-flip produces a coherent superposition across states with different parity. These states do not mix under the quench dynamics and the resulting fidelity cannot be estimated by $F_d$, if the readout basis is the computational basis. In contrast, if we measured in the $X$ basis, $F_d$ would be accurate, as discussed in (Section~\ref{sec:symmetry_nonergodicity}). If instead we had flipped two spins, our state would remain in the even-parity sector, and $F_d$ would estimate $F$ with any readout basis.

In all three cases, we compute the fidelity $F = \abs{\braket{\Psi(\phi^*)}{\Psi(\phi)}}^2 = \cos(\phi-\phi^*)^2$. In Fig.~\ref{fig:target_state_benchmarking_local}(a,c,e), we show that $F_d$, evaluated at a fixed point in time, estimates the fidelity $F$ in all three cases. In Fig.~\ref{fig:target_state_benchmarking_local}(b,d,f), we show that $F_d$ takes a short delay time to learn the value of $F$ to a target state $\phi = \frac{\pi}{2}$ --- only a relatively short quench evolution of $t = O(1/J)$ is needed for this protocol. In the trapped ion model, $F_d$ shows slightly larger deviations from $F$. This is because $F_d$ fluctuates in time about $F$: this is indicated by the blue band in Fig.~\ref{fig:target_state_benchmarking_local}(e) and larger fluctuations Fig.~\ref{fig:target_state_benchmarking_local}(f).

\begin{figure}[tbp]
    \centering
    \includegraphics[width=0.95\textwidth]{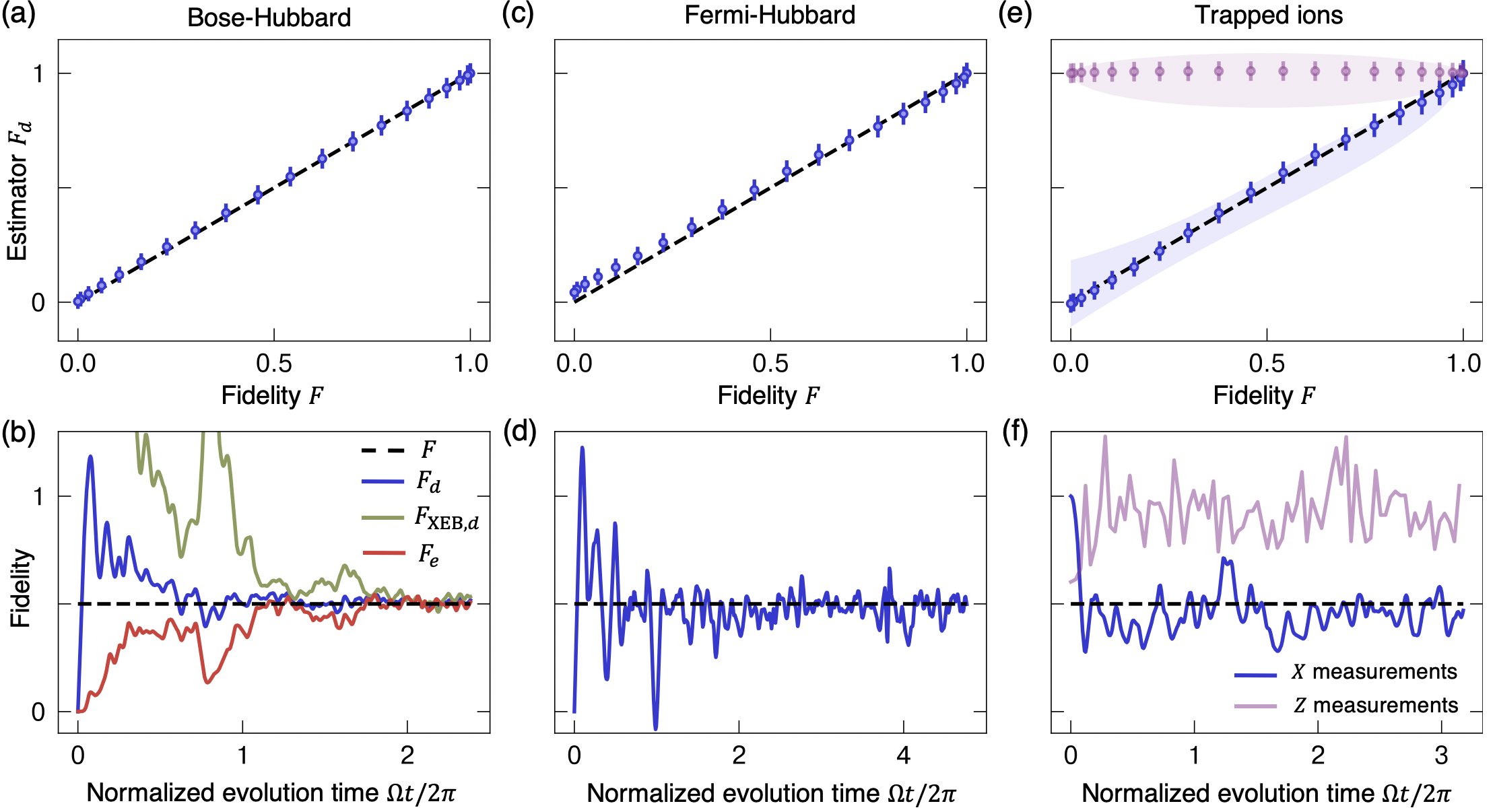}
    \caption{Target state benchmarking of locally-perturbed states. We apply local perturbations to initial states in a (a,b) Bose-Hubbard, (c,d) Fermi-Hubbard, and (e,f) trapped-ion spin chain model. We consider a class of states parametrized by an angle $\phi$. We take states with $\phi^* = \pi/4$ as the simulated prepared state, and compare it to target states with variable $\phi$. As a function of $\phi$, the fidelity between the prepared and target states vary from 0 to 1. By quench evolving these states, we estimate their fidelity with $F_d$. In the top row, we demonstrate that the estimator $F_d$ at a fixed, sufficiently late time, faithfully reproduces the fidelity $F$. Error bars denote the statistical uncertainty from a finite number $M=1000$ of samples. In the bottom row, we demonstrate the required evolution time for $F_d$ to learn the true fidelity $F$, with the target state at $\phi = \pi/2$. In (b), we compare the different estimators discussed in the text: $F_d$ (blue), $F_e$ (green), and $F_\text{XEB,d}$ (red) in a Bose-Hubbard model. The difference between $F_\text{XEB,d}$ and $F_d$ highlights the role of the denominator $\sum_z \pdiagz \tilde{p}(z)^2$ at short times. In (e), we show that evaluating $F_d$ at a fixed time shows slight deviations from $F$. This is due to the small effective dimension $D_\beta \approx 11$ in our trapped ion model: $F_d$ fluctuates about $F$ over time, as seen in (f) and indicated with the blue shaded region in (e). The purple markers and shaded regions in (e,f) indicate readout in the computational basis. Due to constraints from symmetry, $F_d$ does not estimate $F$ (in fact, it is constant in our setting). In (f), the purple curve is evaluated with a target state at $\phi = \pi/2 + \varepsilon$ for a small $\varepsilon$. This is because the target state at $\phi = \pi/2$ occupies only one out of two symmetry sectors, hence is a fine-tuned, pathological case.}
    \label{fig:target_state_benchmarking_local}
\end{figure}

\subsection{Comparing non-locally deformed states}

\begin{figure}[tbp]
    \centering
    \includegraphics[width=0.7\textwidth]{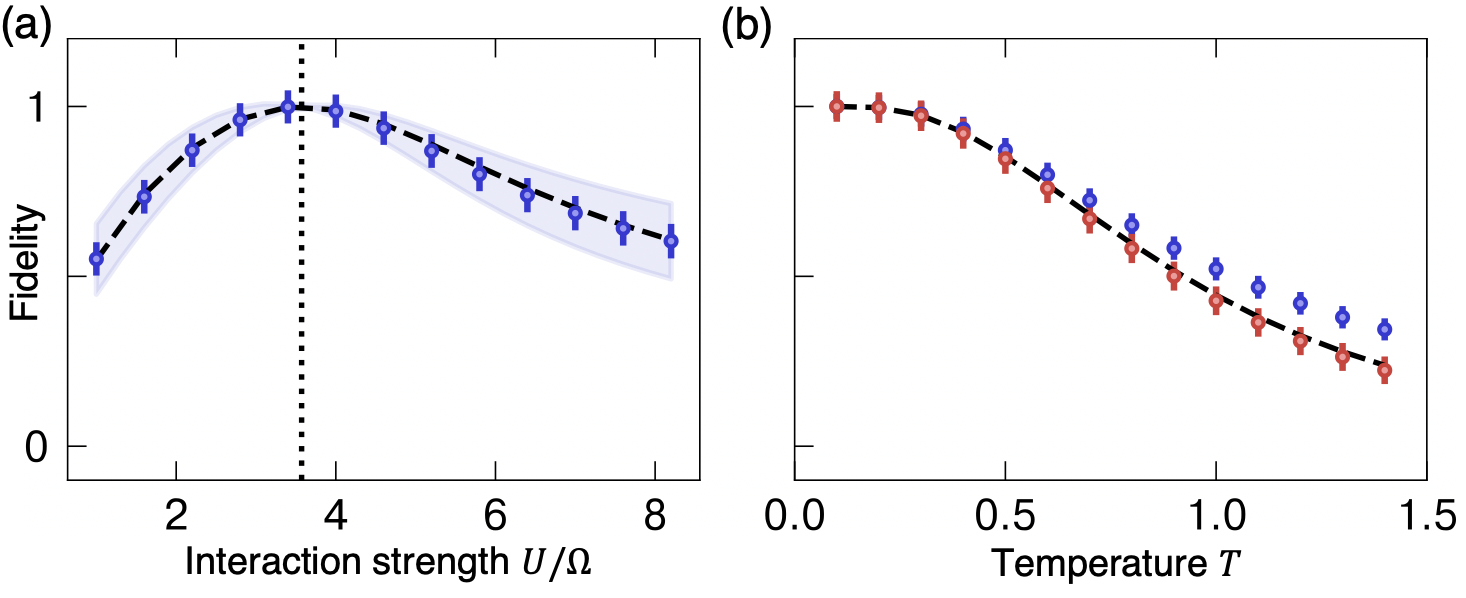}
    \caption{Target state benchmarking by verifying preparation of a ground state in a system. We consider two tasks: (a) estimating the corresponding parameter value and (b) estimating the temperature of the prepared state. In (a), we estimate the fidelity of the prepared state --- assumed to be the ground state at the critical value of $\L(U/\Omega\R)_\text{crit} \approx 3.6$ --- against the ground states at various values of $U/\Omega$. We find that $F_d$ (blue markers) is close to the true fidelity $F$ (dashed line). The blue shaded region indicates the size of fluctuations of $F_d$ over time. In (b), we estimate the fidelity of the prepared state --- assumed to be the Boltzmann state at various temperatures $T$ --- against the ground state at $\L(U/\Omega\R)_\text{crit}$. We find that $F_d$ (blue) slightly over-estimates $F$, while $F_e$ (red) closely estimates $F$, consistent with our analysis in Section.~\ref{app:Fid_FXEB_Fd}. In both plots, error bars indicate the statistical uncertainty associated with 1000 measurement samples.}
    \label{fig:target_state_benchmarking}
\end{figure}

We next demonstrate verifying the preparation of a non-trivial ground state in two ways by computing the fidelities of states related by non-local deformations. We first demonstrate verifying the position in the phase diagram by computing the fidelities between a prepared critical ground state and the ground states of the same Hamiltonian at other parameter values. We then demonstrate estimating the temperature of the prepared state by considering a prepared low-temperature Boltzmann state of the Hamiltonian at the critical point and estimating its fidelity (overlap) to the ground state.

We simulate the ground state of the Bose-Hubbard model \eqref{eq:BH_Ham} tuned to the Mott insulator-superfluid transition point at $\L(U/\Omega\R)_\text{crit} \approx 3.6$ on a system of 10 bosons on 10 lattice sites~\cite{kuhner1998phases}. We wish to benchmark the fidelity of our prepared state against its desired ground state. We consider two purposes: estimating the best parameter $U/\Omega$ of the prepared ground state, and estimating the temperature of the prepared state.

To do so, we consider quench evolving our prepared state via a modified Bose-Hubbard Hamiltonian, and comparing it to simulated data of quenched ground states at various values of $U/\Omega$. Here, care must be taken for two reasons. Firstly, the quench Hamiltonian should be chosen such that the ground states under investigation do not have too low energy with respect to the quench Hamiltonian. Secondly, the family of states we consider do not simply differ by local errors. Rather, the family of ground states are related by possibly non-local transformations. This is outside of the scope of our heuristic derivation. However, we numerically observe that the $F_d$ formula still approximates $F$ under judicious choices of the quench Hamiltonian.

In order for the ground states in both phases to have sufficiently ergodic quench dynamics, we consider a quench with a disordered Hamiltonian: $H = H_\text{BH}\L[(U/\Omega)_\text{crit}\R] + \sum_{j=1}^{8} \Delta_j n_j$, where $\Delta_j \sim \text{Unif}[-3,3]$. In Fig.~\ref{fig:target_state_benchmarking}(a), $F_d$ approximates the true ground state overlap in a system of 8 bosons on 8 lattice sites.

We next demonstrate estimating the temperature of an experimentally prepared state. Assuming we have correctly determined the parameter $U/\Omega$, we simulate the $F_d$ values between the quench-evolved ground state and quench-evolved Boltzmann states $\rho =  Z(T)^{-1} \exp(-H_\text{BH}/T)$, where $Z(T)$ is the partition function at temperature $T$. The fidelity is simply $Z(T)^{-1} \exp(-E_0/T)$. In this case, $F_d$ [Fig.~\ref{fig:target_state_benchmarking}(b), blue], is close to the true fidelity. However, because of a relatively small effective dimension $D_\beta \approx 5$, $F_d$ is larger than $F$, consistent with the prediction in Eq.~\eqref{eq:var_t_prediction}. This systematic shift can be manually subtracted. Alternatively, the estimator $F_e$ [Eq.~\eqref{eq:Frel}] is more robust to small values of $D_\beta$ and closely estimates the fidelity [Fig.~\ref{fig:target_state_benchmarking}(b), red].

\end{widetext}

\end{document}